\documentclass{article}
\usepackage[dvips]{graphicx}
\usepackage{graphicx, epsfig}
\usepackage{bbold}
\usepackage{fullpage}
\usepackage{float}
\usepackage{amsmath}
\usepackage{amsthm}
\usepackage{amssymb}
\usepackage{subfigure}

\newcommand{\bra}[1]{\langle #1|}
\newcommand{\ket}[1]{|#1\rangle}
\newcommand{\braket}[2]{\langle #1|#2\rangle}
\newcommand{\ketbra}[2]{|#1\rangle \langle #2|}
\newcommand{\Tr}{\text{Tr}}
\newtheorem{theorem}{Theorem}

\newtheorem{definition}[theorem]{Definition}

\newtheorem{lemma}[theorem]{Lemma}

\newtheorem{remark}[theorem]{Remark}

\begin{document}
% \onecolumngrid
\title{Trade-off coding for universal qudit cloners motivated by the Unruh effect}

\author{Tomas~Jochym-O'Connor\\
{\it \small Institute for Quantum Computing, Department of Physics and Astronomy},\\{\it \small University of Waterloo, 200~University~Avenue~West, Waterloo,~Ontario, N2L~3G1, Canada}\\
\and Kamil~Br\'adler and Mark~M.~Wilde\\
{\it \small School of Computer Science, McGill University,}\\{\it \small Montreal,~Quebec, H3A~2A7, Canada}}

%\author{Tomas Jochym-O'Connor}
%\email{trjochymoconnor@uwaterloo.ca}
%\affiliation{
%    Institute for Quantum Computing and Department of Physics and Astronomy,
%    University of Waterloo,
%    200 University Avenue West, Waterloo, Ontario, N2L 3G1, Canada
%    }
%\affiliation{
%    School of Computer Science,
%    McGill University,
%    Montreal, Quebec, H3A 2A7, Canada
%    }
%\author{Kamil Br\'adler}
%\email{kbradler@cs.mcgill.ca}
% \affiliation{
%    School of Computer Science,
%    McGill University,
%    Montreal, Quebec, H3A 2A7, Canada
%    }
%\author{Mark M. Wilde}
%    \email{mark.wilde@mcgill.ca}
%\affiliation{
%    School of Computer Science,
%    McGill University,
%    Montreal, Quebec, H3A 2A7, Canada
%    }
\maketitle

% Abstract
\begin{abstract}
A ``triple trade-off" capacity region of a noisy quantum channel provides a more complete description of its capabilities than does a single capacity formula. However, few full descriptions of a channel's ability have been given due to the difficult nature of the calculation of such regions---it may demand an optimization of information-theoretic quantities over an infinite number of channel uses. This work analyzes the $d$-dimensional Unruh channel, a noisy quantum channel which emerges in relativistic quantum information theory. We show that this channel belongs to the class of quantum channels whose capacity region requires an optimization over a single channel use, and as such is tractable. We determine two triple-trade off regions, the quantum dynamic capacity region and the private dynamic capacity region, of the $d$-dimensional Unruh channel. Our results show that the set of achievable rate triples using this coding strategy is larger than the set achieved using a time-sharing strategy. Furthermore, we prove that the Unruh channel has a distinct structure made up of universal qudit cloning channels, thus providing a clear relationship between this relativistic channel and the process of stimulated emission present in quantum optical amplifiers.

\end{abstract}
%\pacs{03.67.Hk, 03.67.Pp, 04.62.+v}

% Introduction
\section{Introduction}

The concept of a noisy quantum channel and the ability to communicate, or rather transfer quantum information, from one party to another using quantum mechanical phenomena has been important in quantum communication. A natural question is to determine the capabilities of a given noisy quantum channel for communicating classical and quantum information at the expense of noiseless entanglement. Protocols such as super-dense coding~\cite{Bennett_SuperDense} and quantum teleportation~\cite{Bennett_Teleportation} have sparked interest in the field of quantum information because they provide concrete examples of the trade-off between classical and quantum resources.

This work expands the study of the trade-off of these resources over a qudit Unruh channel~\cite{Bradler_Rindler, Bradler_Hadamard}. In  the $d$-dimensional case, there is no longer a restriction to encoding qubits with a dual-rail encoding. A sender and receiver can encode the information into a single-excitation $d$-dimensional basis, expanding the physical freedom available for encoding information. The qudit Unruh channel arises in a relativistic setting where a receiver, named Bob, is accelerating uniformly with respect to the reference frame of the sender, named Alice, who is sending encoded information in the form of multi-rail photonic qudits.

We prove that the structure of the $d$-dimensional Unruh channel is directly related to $d$-dimensional universal cloning channels~\cite{Buzek, Gisin, Werner, Keyl, Fan}, as was previously shown in the case of the qubit Unruh channel~\cite{Bradler_Cloning} and the qubit transformation present near black holes, whose metric is locally equivalent to the spacetime of an accelerating observer~\cite{Adami}. Such channels arise from the process of stimulated emission and can occur in physical systems such as the amplification of light in erbium-doped fibers~\cite{Fasel}. This leads to an alternative interpretation of the Unruh channel as a transformation present in optical amplifiers when encoding optical qudits in time-bin photons through an optical fiber with $d$ spatial modes~\cite{deRiedmatten, Thew, Scarani}. Each mode would pass through an optical amplifier and the output state would correspond to that of the Unruh channel where the optical gain of each amplifier now plays the role of the acceleration parameter in the relativistic setting.

A channel's ability to transmit multiple types of information is described by its set of achievable rate tuples. Shor gave the first compact characterization of rate pairs for the transmission of classical information along with the consumption of entanglement~\cite{Shor_CE}, and Devetak and Shor subsequently characterized all achievable rate pairs for classical and quantum communication over a noisy quantum channel~\cite{Devetak_CQ}. Hsieh and Wilde then obtained a characterization of the trade-off between the three resources of classical communication, quantum communication, and entanglement~\cite{Hsieh_CQE}.

While we have formulas for the capacity regions of achievable rate triples, the analysis of these capacity regions can be difficult in practice because it could involve an optimization of entropic formulas over a potentially infinite number of uses of the channel. However, the analysis of these capacity regions becomes tractable when a capacity formula for a quantum channel single-letterizes, that is, when we can evaluate the formulas with respect to one channel use. Some examples of quantum channels that are known to have a single-letter triple trade-off region are the qubit erasure channel, the completely depolarizing channel, the class of Hadamard channels~\cite{Bradler_Hadamard, Hsieh_CQEexamples}.

An analogy can be drawn from the triple trade-off region of quantum resources for a quantum channel to that of the trade-off between the classical resources of public communication, private communication, and shared secret key~\cite{Collins_CAnalogue}. Wilde and Hsieh studied the trade-off formulas for the use of these noiseless resources along with a quantum channel to generate other noiseless resources~\cite{Wilde_PrivateQuantumTO}. The triple trade-off region for a noisy quantum channel combined with these resources is also tractable for the class of Hadamard channels.

In this paper, we prove that universal qudit cloning and Unruh channels belong to the Hadamard class of channels, and thus they belong to the class of known channels whose triple trade-off capacity regions single-letterize. To do so, we first show that the qudit Unruh channel has a particular block diagonal form where each block is a universal qudit cloner. Showing that the Unruh channel is a Hadamard channel can be reduced to showing that a $1\rightarrow 2$ universal qudit cloner is a Hadamard channel. It was shown for $d=2$ that the complementary channel of a $1 \rightarrow 2$ universal qubit cloner is a convex combination of completely dephasing channels~\cite{Bradler_Hadamard}, but in higher dimensions no such direct comparison to known channels was clear. However, a combinatorial argument is given in this work showing that a $1\rightarrow 2$ universal qudit cloner is Hadamard. The resulting triple trade-off capacity region has a richer structure than that of the qubit case because its computation requires a set of parameters that grow with dimension size.

This paper is structured as follows. Section~\ref{sec:background} reviews the mathematical structure of the qudit Unruh channel and the formulas for the various capacity regions that are of interest in our study of this channel. Section~\ref{sec:hadamard} presents one of the main results of the paper, that the qudit Unruh channel is a Hadamard channel, and this is based on a derivation in Appendix~A. 
%\ref{app:entanglementbreaking}%.
This result implies that the capacity formulas for this channel single-letterize~\cite{Bradler_Hadamard}, and it thus enables the calculation of the capacity regions. In Section~\ref{sec:capacity_regions}, we provide plots of the various capacity regions and show how a trade-off coding strategy beats time-sharing. We conclude in Section~\ref{sec:conclusion} with some final remarks and present some open questions in relativistic quantum information theory.

% Background section on channel structure and information-theoretic quantities
\section{Background and Review}
\label{sec:background}

\subsection{Channel Structure}
\label{sec:channel_structure}
The qudit Unruh channel has appeared in a relativistic setting where two inertial observers Alice and Bob, henceforth denoted $A$ and $B$, are communicating via the exchange of bosons through a quantum channel, and an eavesdropper Eve, denoted $E$, is accelerating uniformly with respect to Alice and Bob's reference frame~\cite{Bradler_Rindler}. However, in this paper we consider a variant of the above situation with the receiver, Bob, now accelerating uniformly with respect to the sender, Alice, and the environment of the channel is represented by Eve.

%
%
% CORRECTION: FIX LABELING ON INPUT STATES (EQS. 1-6)
%
%
\begin{definition}[\cite{Bradler_Rindler}]
The qudit Unruh channel $\mathcal{N}^{A \rightarrow B}$ is defined by the following map:
\begin{align}
\mathcal{N}^{A \rightarrow B}(\rho_{A}) = \Tr_E (\mathcal{U} \rho_{A}\otimes \ketbra{\text{vac}}{\text{vac}}_E \mathcal{U}^{\dagger}),
\end{align}
where $\rho_{AE}$ is the input state of the $A$ subsystem combined with that of the environment $E$, and $\mathcal{U}=\otimes_{i=1}^d \mathcal{U}_{A_i E_i}$, such that
\begin{align}
\mathcal{U}_{A_i E_i}(r) = \dfrac{1}{\cosh{r}} \exp[ \tanh{r} a_i^{\dagger} e_i^{\dagger}] \times \exp[ - \ln{ \cosh{r}(a_i^{\dagger}a_i + e_i^{\dagger} e_i) }  ] \times \exp[ -\tanh{r} a_i e_i] ,
\end{align}
where $a_i^{\dagger}$ and $a_i$ are the respective creation and annihilation operators for the $A$ subsystem, $e_i^{\dagger}$ and $e_i$ are the respective creation and annihilation operators for the $E$ subsystem, and $r$~characterizes the acceleration of the receiver.
\end{definition}

The following commutation relations hold for the case of bosonic creation and annihilation operators: $[a_i, a_j^{\dagger}] = \delta_{ij}$, $[a_i^{\dagger},a_j^{\dagger}] = 0$,  $[a_i^{\dagger},e_j^{\dagger}] = 0$,  $[a_i^{\dagger},e_j] = 0$. The quantum states we consider in this work are photonic mode states expressed in a multi-rail basis. 
%
%
% CORRECTION (change the term "d-dimensional" to "d-mode")
%
%
We represent a $d$-mode photonic state as $\ket{(n_1, \cdots, n_d)}$, where the value $n_i$ denotes the number of photons in the $i$-th mode. The creation and annihilation operators then act as follows on these states:
\begin{align}
a_i^{\dagger}\ket{(n_1,\cdots, n_i, \cdots, n_d)}& = \sqrt{1+n_i} \ket{(n_1,\cdots, n_i+1, \cdots, n_d)},\nonumber \\
a_i\ket{(n_1,\cdots, n_i, \cdots, n_d)} &= \sqrt{n_i} \ket{(n_1,\cdots, n_i-1, \cdots, n_d)}. \nonumber
\end{align}

%
%
% CORRECTION (Added a little explanation on the definition of the qudit Unruh channel)
%
%
The input state we consider in this work has the form $\ket{\psi}_{AE} = \sum_{i=1}^d \beta_i a_i^{\dagger} \ket{\text{vac}}_A \ket{\text{vac}}_E$, that is, Alice's input state is a single-excitation qudit. It is for this reason that we refer to the channel of interest as a qudit ($d$-dimensional) Unruh channel, even though the output state is infinite dimensional. As shown in~\cite{Bradler_Rindler}, the following identity holds for an input of the above form:
%$\ket{\psi}_{AE} = \sum_{i=1}^d \beta_i a_i^{\dagger} \ket{\text{vac}}_A \ket{\text{vac}}_E$:
\begin{align}
\ket{\sigma}_{BE} &= \mathcal{U} \ket{\psi}_{AE} = \mathcal{U} \sum_{i=1}^d \beta_i \ket{(0, \cdots, 0,1_i, 0, \cdots, 0)}_A \ket{\text{vac}}_E \nonumber \\
& = \dfrac{1}{\cosh^{d+1}{r}} \sum_{k=1}^{\infty} \tanh^{k-1}{r} \sum_{N(k-1)} \sum_{i=1}^d \sqrt{1+n_i} \beta_i \ket{(n_1, \cdots, 1+n_i, \cdots, n_d)}_B  \ket{(n_1,\cdots, n_d)}_E,
\end{align}
%
%
% CORRECTION: Added in some text about what the set N(k) represents
%
%
where $N(k) = \{ n_i, \text{ }1 \le i \le d \text{ } | \sum_{i=1}^d n_i = k \}$ and the states $\ket{.}_B$ and $\ket{.}_E$ are expressed in a multi-rail basis. The set of states that fall into $N(k)$ is the set of $d$-mode states with a total number of $k$~photons. One should note that after the action of $\mathcal{U}$ we have relabelled the $A$ subsystem as the $B$ subsystem as this is the state Bob receives.

Thus if one inputs a pure state of the form $\psi_{AE} = \ketbra{\psi}{\psi}$ then the output of the qudit Unruh channel has the following form:
\begin{align}
\label{eq:Unruh_blocks}
\sigma_B = (1-z)^{d+1} \bigoplus_{k=1}^{\infty} z^{k-1} \sigma_B^{(k)},
\end{align}
where $z=\tanh^2{r}$ and $r$ is the acceleration parameter. The blocks in this infinite-dimensional matrix are given by
\begin{align}
\sigma_B^{(k)}=\sum_{N(k-1)}  \Big(  \sum_{i=1}^d \beta_i \sqrt{1+n_i} & \ket{(n_1, \cdots, 1+n_i, \cdots, n_d)}_B\Big) \Big( \sum_{j=1}^d \overline{\beta}_j \sqrt{1+n_j}  \bra{(n_1, \cdots, 1+n_j, \cdots, n_d)}_B \Big) 
\end{align}
%
%
% Correction: Added in a sentence about the dimension of the k-th block
%
%
where $N(k)$ is again defined as above, reflecting that the $k$-th block of Bob's state represents states with $k$~photons. The dimension of the $k$-th block is $ {k+d-1 \choose d-1}$.

Finally, the complementary channel to the Unruh channel $\mathcal{N^C}$ is calculated by tracing out Bob's subsystem rather than Eve's:
\begin{align}
\mathcal{N^C}^{A \rightarrow E}(\rho_A) = \Tr_B (\mathcal{U} \rho_{A}\otimes \ketbra{\text{vac}}{\text{vac}}_E \mathcal{U}^{\dagger}).
\end{align}

% Subsection on Cloning Channels.
\subsection{Qudit cloning channels}
While the \textit{No-cloning Theorem} states that there is no unitary transformation that can copy an arbitrary quantum state~\cite{No-cloning}, there is a notion of the best possible approximate copy of a qudit. The unitary transformation that performs this approximate copy is referred to as a \textit{universal qudit cloning machine}~\cite{Buzek, Gisin, Werner, Keyl}. We define the universal cloning channel as a completely positive map given by tracing over the auxiliary output of the universal cloning machine. As shall be shown in Section~\ref{sec:hadamard}, these universal qudit cloners are the building blocks for the qudit Unruh channel and will be used extensively in the calculation of the capacity regions in Section~\ref{sec:capacity_regions}.

%
%
% CORRECTION (changes to definition)
%
%
\begin{definition}[\cite{Fan}]
\label{def:cloner}
Let $\ket{\psi}$ be an arbitrary qudit $\ket{\psi} = \sum_{i=1}^d x_i \ket{i}$ such that $\sum_{i=1}^d |x_i|^2 = 1$. The $N$ to $M$ universal qudit cloner, $M\ge N$, is a unitary transformation that takes an input of the form
\begin{align}
\ket{\psi}^{\otimes N} \otimes R = \sum_{\vec{n}} \sqrt{ \dfrac{N!}{n_1! \cdots n_d!} } x_1^{n_1} \cdots x_d^{n_d} \ket{\vec{n}} \otimes R,
\end{align}
where the sum $\sum_{\vec{n}}$ is the sum over vectors $\ket{\vec{n}} = \ket{n_1,\cdots, n_d}$ which are completely symmetric normalized states with $n_i$ systems in state $\ket{i}$ such that $\sum_{i=1}^d n_i = N$. $R$ represents the state of an auxiliary subsystem. The unitary cloning operation $\mathcal{U}_{NM}$ is given by the following transformation
\begin{align}
\label{eq:cloner_output}
\mathcal{U}_{NM}\ket{\vec{n}} \otimes R = \sum_{\vec{j}} \alpha_{\vec{n} \vec{j} } \ket{ \vec{n} + \vec{j} } \otimes R_{\vec{j}},
\end{align}
where the sum $\sum_{\vec{j}}$ is over all $\vec{j}$ such that $\sum_{i=1}^d j_i = M-N$ and the states  $R_{\vec{j}}$ of the auxiliary subsystem are orthogonal. Finally the coefficients $\alpha_{\vec{n} \vec{j} }$ are given by
\begin{align}
\label{eq:cloner_eigs}
\alpha_{\vec{n} \vec{j} } = \sqrt{ \dfrac{(M-N)! (N+d-1)!}{(M+d-1)!} } \sqrt{ \prod_{i=1}^d \dfrac{(n_i + j_i)!} {n_i! j_i!} } .
\end{align}
\end{definition}

\subsection{Classical and Quantum Information Quantities}
\label{sec:capacities}
In this section we review some of the important information-theoretic quantities for the study of trade-offs between different protocols.

\subsubsection{Classically-Enhanced Quantum Capacity Region}
The classically-enhanced quantum capacity region is the region of achievable rate pairs that characterize the ability of a noisy quantum channel $\mathcal{N}$ to communicate classical ($C$) and quantum ($Q$) information. Devetak and Shor showed that the capacity region for a noisy quantum channel $\mathcal{N}$ is as follows~\cite{Devetak_CQ}:
\begin{align}
\mathcal{C}_{CQ}(\mathcal{N}) \equiv \overline{ \bigcup_{k=1}^{\infty} \dfrac{1}{k} \mathcal{C}_{CQ}^{(1)}(\mathcal{N}^{\otimes k}) },
\label{eq:multiletterCQ}
\end{align}
where the so-called ``one-shot" region $ \mathcal{C}_{CQ}^{(1)}(\mathcal{N}) $ is defined by
\begin{align}
\mathcal{C}_{\text{CQ}}^{(1)}(\mathcal{N})\equiv\bigcup_{\rho}\mathcal{C}_{\text{CQ,}\rho}^{(1)}(\mathcal{N}).\nonumber
\end{align}
The ``one-shot, one-state" region $\mathcal{C}_{\text{CQ,}\rho}^{(1)}(\mathcal{N})$ is the set of all $C,Q \ge 0$ that satisfy the following inequalities:
\begin{align}
C  &  \leq I(X;B)_{\rho},\nonumber\\
Q  &  \leq I(A\rangle BX)_{\rho},\nonumber
\end{align}
where the state $\rho$ has the form
\begin{align}
\rho^{XABE}\equiv\sum_{x}p_{X}\left(  x\right)  \left\vert x\right\rangle
\left\langle x\right\vert ^{X}\otimes U_{\mathcal{N}}^{A^{\prime}\rightarrow
BE}(\phi_{x}^{AA^{\prime}})
\label{eq:CEQ-state}%
\end{align}
and the states $\phi_x^{AA'}$ are pure.

\subsubsection{Entanglement-Assisted Classical Capacity Region}
The entanglement-assisted classical capacity region characterizes the ability for a noisy quantum channel $\mathcal{N}$ to transmit classical information with the help of noiseless quantum entanglement. The capacity region is the set of all achievable pairs ($C$,$E$)  where $C$ is the rate of classical communication and $E$ is the rate of entanglement consumption. Shor showed that this region is characterized by the following~\cite{Shor_CE}:

\begin{align}
\mathcal{C}_{CE}(\mathcal{N}) \equiv \overline{ \bigcup_{k=1}^{\infty} \dfrac{1}{k} \mathcal{C}_{CE}^{(1)}(\mathcal{N}^{\otimes k}) }
\label{eq:multiletterCE}
\end{align}
where the ``one-shot" region $\mathcal{C}_{CE}^{(1)}(\mathcal{N})$ is defined in a similar way as $\mathcal{C}_{CQ}^{(1)}(\mathcal{N})$. The union is now taken over the set of ``one-shot, one-state" regions $\mathcal{C}_{CE, \rho}^{(1)}(\mathcal{N})$, that are defined by the inequalities
\begin{align}
C &\le I(AX;B)_{\rho}, \\
E &\ge H(A|X)_{\rho},
\end{align}
and the state $\rho$ is defined as in~\eqref{eq:CEQ-state}.

% Subsection
\subsubsection{Quantum Dynamic Capacity Region}
Finally, the capacity region for entanglement-assisted transmission of quantum and classical information (CQE) is characterized by the following region~\cite{Hsieh_CQE}:
\begin{align}
\mathcal{C}_{CQE}(\mathcal{N}) \equiv \overline{ \bigcup_{k=1}^{\infty} \dfrac{1}{k} \mathcal{C}_{CQE}^{(1)}(\mathcal{N}^{\otimes k}) },
\label{eq:multiletterCQE}
\end{align}
where again the ``one-shot" region $\mathcal{C}_{CQE}^{(1)}(\mathcal{N}) $
is defined to be the union of ``one-shot, one-state" regions $\mathcal{C}_{CQE,\rho}^{(1)}(\mathcal{N}) $, defined by the following inequalities:
\begin{align}
C + 2Q &\le I(AX;B)_{\rho}, \\
Q + E &\le I(A \rangle BX)_{\rho}, \\
C+Q +E &\le I(X;B)_{\rho} + I(A \rangle BX)_{\rho},
\end{align}
where again the state $\rho$ is defined as in~\eqref{eq:CEQ-state}.

% Define Public-Private-Secret-Key Region
\subsection{Private Dynamic Capacity Region}
\label{sec:PDC_region}
In the previous section we reviewed the capabilities of a noisy quantum channel to transmit classical information and quantum information at the expense of noiseless entanglement. However, one may also study the trade-off of the classical resources of noiseless public communication, private communication, and secret key. Wilde and Hsieh provided the following characterization of the private dynamic capacity region, $\mathcal{C}_{RPS}(\mathcal{N})$, of a quantum channel $\mathcal{N}$~\cite{Wilde_PrivateQuantumTO}:

\begin{align}
\mathcal{C}_{RPS}(\mathcal{N}) \equiv \overline{ \bigcup_{k=1}^{\infty} \dfrac{1}{k} \mathcal{C}_{RPS}^{(1)}(\mathcal{N}^{\otimes k}) }
\label{eq:multiletterRPS}
\end{align}
where the ``one-shot" region $\mathcal{C}_{RPS}^{(1)}(\mathcal{N}) $
is the union of ``one-shot, one-state" regions $\mathcal{C}_{RPS,\rho}^{(1)}(\mathcal{N}) $:
\begin{align}
\mathcal{C}_{RPS}^{(1)}(\mathcal{N}) = \bigcup_{\rho} \mathcal{C}_{RPS, \rho}^{(1)}(\mathcal{N}).
\end{align}
The ``one-shot, one state" region $\mathcal{C}_{RPS, \rho}^{(1)}(\mathcal{N})$ is defined by a set of inequalities over the rates of public classical communication $R$, private classical communication $P$, and secret key generation $S$:
\begin{align}
R + P&\le I(YX;B)_{\rho}, \\
P + S &\le I(Y; B|X)_{\rho} - I(Y;E|X)_{\rho}, \\
R+P +S &\le I(YX;B)_{\rho} - I(Y;E|X)_{\rho},
\end{align}
where the state $\rho$ is defined as
\begin{align}
\rho^{XYBE} = \sum_{x,y} p_{X,Y} (x,y) \ketbra{x}{x}^X \otimes \ketbra{y}{y}^Y \otimes \mathcal{U}_{\mathcal{N}}^{A' \rightarrow BE} (\rho_{x,y}^{A'}).
\end{align}

% Section
\section{The Hadamard class of quantum channels}
\label{sec:hadamard}
The calculation of the information capacity regions outlined in Section~\ref{sec:capacities} can be difficult because they generally require taking a union over an infinite number of uses of the channel. However, there are certain  classes of channels such that the full capacity region of entanglement-assisted transmission of quantum and classical information ``single-letterizes", that is, $\mathcal{C}_{CQE}(\mathcal{N}) = \mathcal{C}_{CQE}^{(1)} (\mathcal{N})$. Among these is the class of Hadamard channels~\cite{Bradler_Hadamard}.
\begin{definition}
A noisy quantum channel $\mathcal{N}^{A' \rightarrow B}$ is a Hadamard channel if its complementary channel $\mathcal{\left(N^{C}\right)}^{A' \rightarrow E}$ is entanglement breaking~\cite{King_EB, King_Hadamard}. That is, the action of the complementary channel on an entangled state $\ket{\psi}^{AA'} = \dfrac{1}{\sqrt{d}} \sum_{i=1}^d \ket{i}^A \ket{i}^{A'}$ is as follows:
\begin{align}
\mathcal{\left(N^C\right)}^{A' \rightarrow E} (\ketbra{\psi}{\psi}^{AA'}) = \sum_{x \in X} p_X (x) \rho_{x}^A \otimes \sigma_x^E.
\end{align}
\end{definition}

% Subsection concerning the comparison between the Unruh channel and qudit cloners.
% \subsection{Comparing the qudit Unruh channel to qudit cloners}
In order to show that the Unruh channel belongs to the class of Hadamard channels, we first show that each of its blocks in~\eqref{eq:Unruh_blocks} are universal qudit cloning channels. In Section~\ref{sec:capacity_regions}, we shall be calculating the information capacities of the universal qudit cloners and then using the knowledge of the following lemma to simplify the calculation of the capacity region of the qudit Unruh channel.

%
%
% CORRECTION: Addition of the definition of covariant
%
%
\begin{definition}[\cite{Bradler_Rindler}]
Let $G$ be a group, $\mathcal{H}_{in}$, $\mathcal{H}_{out}$ be Hilbert spaces and let $r_1~:~G \rightarrow GL(\mathcal{H}_{in})$, $r_2~:~G \rightarrow GL(\mathcal{H}_{out})$ be unitary representations of the group, where $GL$ stands for general linear representation. Let $\mathcal{K} : \mathcal{DM}(\mathcal{H}_{in}) \rightarrow \mathcal{DM}(\mathcal{H}_{out})$ be a channel, where $\mathcal{DM}$ stands for the space of density matricies. The channel $\mathcal{K}$ is defined to be covariant with respect to $G$, if
\begin{align}
\mathcal{K} \big( r_1(g) \rho r_1(g)^{\dagger} \big) = r_2(g) \mathcal{K}(\rho) r_2(g)^{\dagger}
\end{align}
holds for all $g \in G$,  $\rho \in \mathcal{DM}(\mathcal{H}_{in})$.
\end{definition}

\begin{remark}
As shown in Ref.~\cite{Bradler_Rindler}, the qudit Unruh channel is $SU(d)$-covariant. Thus, a unitary transformation upon a given input state will result in an output state that is equivalent to the original output state of the qudit Unruh channel up to a unitary transformation.
\end{remark}

% Lemma
\begin{lemma}
\label{lem:blockcloner}
The $k$-th block of the qudit Unruh channel, after normalization of the block, is equivalent to a $1 \rightarrow k$ universal qudit cloner.
\end{lemma}

% Proof
\begin{proof}
We shall show that both channels share the same eigenvalues with the same multiplicity for each eigenvalue. Consider a particular block of the Unruh channel ($d,k$). This block corresponds to
\begin{align}
\label{eq:output_block}
\ket{\sigma^{(k)} }_{BE}= \dfrac{1}{\sqrt{M}} \sum_{N(k-1)} \sum_{i=1}^d \sqrt{1+n_i} \beta_i \ket{(n_1, \cdots, 1+n_i, \cdots, n_d)}_B \ket{(n_1,\cdots, n_d)}_E,
\end{align}
where $M$ is the normalization coefficient for this block, calculated below, and the same convention as before is used for $N(k)$. 
%
%
% Correction: Define the use of the term covariant
%
%
%As shown in Ref.~\cite{Bradler_Rindler}, the Unruh channel is covariant, that is a unitary transformation to the input state will result in an output state that is equivalent to the original output state up to a unitary transformation.
%Since the Unruh channel is covariant~\cite{Bradler_Rindler},
Since the qudit Unruh channel is $SU(d)$-covariant, any choice of a input qudit pure state will produce the same eigenvalues for the resulting output density matrix. In this case, it will be convenient to fix the input state to be equal to $\ket{1}$, thus $\beta_1=1$ and $\beta_{j\ge2}=0$. Therefore, in our case, the form of the output density operator is as follows:
\begin{align}
\sigma_B^{(k)}&=\dfrac{1}{M} \sum_{N(k)} \sum_{i=1}^d |\beta_i|^2 n_i \ketbra{(n_1, \cdots, n_d)}{(n_1, \cdots, n_d)}_B \\
&=\dfrac{1}{M} \sum_{N(k)}  n_1 \ketbra{(n_1,\cdots, n_d)}{(n_1,\cdots,n_d)}_B
\end{align}
where $M$ is the normalization factor. Since we are considering the block with fixed $(d,k)$, the eigenvalues of the output density operator for Bob are $1 \le n_1 \le k$. Now we count the multiplicities of each eigenvalue. If we fix $n_1=p$ then we have the constraint
\begin{align}
\label{eq:constraint_ns}
\sum_{i=2}^d n_i = k-p
\end{align}
since the sum of all the indices must equal $k$. Now we are free to choose the $n_i$'s such that \eqref{eq:constraint_ns} holds. The multiplicity is then just the number of ways one can choose $n_2, \cdots,n_d$ such that \eqref{eq:constraint_ns} holds, which is equal to
\begin{align}
\label{eq:multiplicity_unruh}
m_p= {(k-p)+d-2 \choose d-2}.
\end{align}
Finally, normalizing the output density operator to have trace one, we have to normalize the eigenvalues by
\begin{align}
\label{eq:normalization_unruh}
M=\sum_{p=1}^d p m_p = { k+d-1 \choose d }.
\end{align}

%
%
% Correction: arbitrary input "qudit pure" state (was added to avoid confusion with just any arbitrary state
%
%
Now consider the eigenvalues for the $1\rightarrow k$ universal qudit cloner in $d$-dimensions. They are given by $\alpha_{\vec{n},\vec{j}}^2$ in~\eqref{eq:cloner_eigs}. By construction~\cite{Gisin, Werner, Keyl}, the universal qudit cloner is a covariant channel, and therefore we can consider an arbitrary input qudit pure state, which we shall choose to be $\ketbra{1}{1}_A=\ketbra{(1,0,\cdots,0)}{(1,0,\cdots,0)}$. Thus now the eigenvalues take the following form~\cite{Fan}:
\begin{align}
\alpha_{1,\vec{j}}^2=\dfrac{1}{ { k+d-1 \choose d }} (1+j_1).
\end{align}
Since the qudit cloner must satisfy the condition $\sum_{i=1}^d n_i+j_i = k$, we must have $0 \le j_1 \le k-1$, or, $ \alpha_{1,\vec{j}}^2 = \dfrac{i}{M},$ $1 \le i \le k$ where $M$ is equal to the normalization factor found in~\eqref{eq:normalization_unruh}. Now we consider the multiplicity of each eigenvalue. Say we consider a particular eigenvalue where $\alpha_{1, \vec{j'}}=\dfrac{b}{M}$. Then $j'_1=b-1$ and the remaining indices are chosen such that $\sum_{i=2}^d j'_i = k-b$. Thus considering all possibilities for the remaining indices of the vector $j'$, the multiplicity of such an eigenvalue is
\begin{align}
m_b = { (k-b)+d-2 \choose d-2}
\end{align}
which agrees with the multiplicity found in~\eqref{eq:multiplicity_unruh}.
\end{proof}

%
%
% Correction: Added Remark about how the states Lemma 6 are completely symmetric
%
%
\begin{remark}
It is worth clarifying that the states in~\eqref{eq:output_block} are in fact a representation of completely symmetric states as defined in Definition~\ref{def:cloner} and required according to~\eqref{eq:cloner_output}. These physical photonic Fock states can be thought of as completely symmetric since there always exists an isomorphism between the 
%Schwinger representation (that is, the representation of the photonic multi-rail basis states as defined by the creation and annihilation operators 
bosonic operator representation as defined and studied in~\cite{Bradler_Rindler} and the subspace of completely symmetric states. Thus, such a mathematical equivalence between subspaces is sufficient to claim that the states in~\eqref{eq:output_block} are indeed a representation of completely symmetric states. Bob's output state will be a mixture of completely symmetric states, but the overall state itself will not be completely symmetric. The most important aspect of the proof is to remark that the output states of both the $1 \rightarrow k$ qudit cloning channel and the $k$-th block of the qudit Unruh channel are both diagonal with the same eigenvalues and multiplicities. Thus, the resulting Kraus operators from both channels must be identical and therefore, by definition, these channels represent the same completely positive trace-preserving (CPTP) map.
\end{remark}

We now proceed to showing that the qudit Unruh channel is among the class of Hadamard channels for all $d$. This key result will enable the calculation of the full capacity region in Section~\ref{sec:capacity_regions} due to the ``single-letterization" of the quantum dynamic capacity formula~\cite{Wilde_QDC} for Hadamard channels. Theorem~\ref{thm:Unruh_Hadamard} below is a consequence of Theorem~\ref{thm:Horodecki} characterizing entanglement-breaking channels and Lemma~\ref{lem:entanglementbreaking} that explicitly constructs rank-one Kraus operators for complementary maps of all $1\to 2$ qudit cloners. The proof of Lemma~\ref{lem:entanglementbreaking} can be found in Appendix~A %\ref{app:entanglementbreaking}%
There is an alternative method to showing the Unruh channel is Hadamard based on results from Vollbrecht and Werner~\cite{Vollbrecht} by investigating the properties of the Jamio\l kowski representation of complementary channels to optimal $1\to2$ qudit cloners. The difference in the proof given in this work is that it is constructive and explicitly shows the entanglement-breaking character of the complementary channels.

\begin{theorem}[\cite{Horodecki}]
\label{thm:Horodecki}
A noisy quantum channel $\mathcal{N}$ is entanglement-breaking if and only if it can be written as a sum of rank-one Kraus operators, that is
\begin{align}
\mathcal{N}(\rho) &= \sum_i A_i \rho A_i^{\dagger} \\
I &= \sum_i A_i^{\dagger} A_i
\end{align}
where each $A_i$ is a rank-one operator.
\end{theorem}

\begin{lemma}\label{lem:entanglementbreaking}
The following operators form a set of rank-one Kraus operators for the first block of the complementary channel of the $1 \rightarrow 2$ universal qudit cloner $\mathcal{S}_2^{(d)}$, in $d$-dimensions,
\begin{align*}
&\dfrac{1}{\sqrt{d+1}}\ketbra{1}{1},\cdots,\dfrac{1}{\sqrt{d+1}}\ketbra{d}{d}, \dfrac{1}{\sqrt{4^{d-1}(d+1)}}\ketbra{\psi(\vec{\textbf{n}})}{\psi(\vec{\textbf{n}})}\sigma_z(\vec{\textbf{n}}), \\
 \text{where } \ket{\psi(\vec{\textbf{n}})}&=\sum_{j=1}^d i^{n_j}\ket{j}, \text{ } \sigma_z(\vec{\textbf{n}})=\sum_{j=1}^d (-1)^{n_j}\ketbra{j}{j},  \text{ }\vec{\textbf{n}}=(n_1,\cdots,n_d),\text{ } n_1=0,\text{ }n_{j\ge2} \in \{0,1,2,3\} .
\end{align*}
\end{lemma}

\begin{theorem}
\label{thm:Unruh_Hadamard}
The qudit Unruh channel is a Hadamard channel.
\end{theorem}

\begin{proof}
We begin by considering the action of a $1 \rightarrow k$ universal qudit cloner acting on an arbitrary qudit state. As shown in~\cite{Bradler_Rindler}, if we write the input qudit state in terms of the generators of the sl($d,\mathbb{C}$) algebra, then the action of the universal qudit cloner is given by the following transformation:
\begin{equation}
    Cl_{1\rightarrow k}^{(d)}(\sigma_A^{(1)})=   Cl_{1\rightarrow k}^{(d)} \left( \dfrac{1}{d}\Big(\mathbb{I}+\sum_{\alpha=1}^{L}m_\alpha\lambda^{(1)}_\alpha\Big) \right)
   ={1\over d}\Big(k\mathbb{I}+\sum_{\alpha=1}^{L}m_\alpha\lambda^{(k)}_\alpha\Big) = \sigma_B^{(k)},
\end{equation}
where each $\lambda_\alpha^{(k)}$ is a generator of the $k^{\text{th}}$ completely symmetric representation of the sl($d, \mathbb{C}$) algebra. The action of the channel can easily be expressed~as

\begin{equation}\label{gentrans}
 Cl_{1\rightarrow k}^{(d)}(m_\alpha\lambda^{(1)}_\alpha) = m_\alpha\lambda^{(k)}_\alpha.
\end{equation}
The universal $1\to k$ qudit cloning channel is a unital channel that has the property that it maps generators of the fundamental representation to generators of higher dimensional completely symmetric representations~\cite{Bradler_Rindler}, in this case the $k^{\text{th}}$ completely symmetric representation.

Consider now the complementary channel $\mathcal{S}_k^{(d)}$ of the $ 1 \rightarrow k$ universal qudit cloner. The action of the channel is described in terms of these generators as follows~\cite{Bradler_Rindler}:
\begin{equation}
    \mathcal{S}_k^{(d)}(\sigma_A^{(1)})= \mathcal{S}_k^{(d)} \left( {1\over d}\Big(\mathbb{I}+\sum_{\alpha=1}^{L}m_\alpha\lambda^{(1)}_\alpha\Big) \right)
    ={1\over d}\Big((k-1)\mathbb{I}+\sum_{\alpha=1}^{L} \overline{ m}_\alpha \lambda^{(k-1)}_\alpha\Big)+\mathbb{I} = \sigma_E^{(k)}
\end{equation}
since
\begin{equation}\label{gentrans1}
\mathcal{S}_k^{(d)}(m_\alpha\lambda^{(1)}_\alpha) =  \overline{ m}_\alpha \lambda^{(k-1)}_\alpha,
\end{equation}
where $\overline{m}_{\alpha}$ denotes the complex conjugate of the coefficient $m_{\alpha}$. We now proceed to showing that $\mathcal{S}_k^{(d)}$ is entanglement-breaking for all $k$, that is
$\varsigma_{BE}=(\mathbb{I} \otimes\ \mathcal{S}_k^{(d)})(\Phi^+)$
is separable. Expressing $\varsigma_{BE}$ in terms of fundamental and higher-dimensional generators, we obtain the following expression:
\begin{align}
\label{eq:varsigma2}
\varsigma_{BE}=\sum_{i=1}^d \sum_{j=1}^d \ketbra{i}{j} \otimes \mathcal{S}_k^{(d)}(\ketbra{i}{j})
&= \sum_{i=1}^d \sum_{j=1}^d  \ketbra{i}{j} \otimes  \mathcal{S}_k^{(d)} \left( \dfrac{1}{d} ( \mathbb{I} + \sum_{\alpha=1}^L m_{\alpha, ij} \lambda_{\alpha}^{(1)}) \right) \nonumber \\
&= \sum_{i=1}^d \sum_{j=1}^d \ketbra{i}{j} \otimes  \dfrac{1}{d} ((d+k-1) \mathbb{I} + \sum_{\alpha=1}^L \overline{m}_{\alpha, ij}  \lambda_{\alpha}^{(k-1)}).
\end{align}

Now we can use both Theorem~\ref{thm:Horodecki} and Lemma~\ref{lem:entanglementbreaking} that demonstrate the complementary channel is entanglement-breaking for the case where $k=2$,
\begin{equation}\label{eq:varsigma3}
\varsigma_{BE}=\sum_l q_l \chi_l \otimes\xi_l
=\sum_l q_l \chi_l
\otimes{1\over d}\Big((k+d-1)\mathbb{I}+\sum_{\alpha=1}\overline{ n}_{\alpha, l}\lambda_{\alpha, l}^{(k-1)}\Big).
\end{equation}
Thus there exists a mapping from the coefficients $\{ m_{\alpha,ij} \}$ in~\eqref{eq:varsigma2} to the set of coefficients $\{ q_l, n_{\alpha, l} \}$ and density matricies $\{ \chi_l \}$ in~\eqref{eq:varsigma3}. Moreover, since the action of the channel is in fact identical in structure for all $k$, in that it just maps to higher-dimensional algebra generators for higher $k$, then this mapping should exist for all $k$. Thus, the result in~\eqref{eq:varsigma3} should exist for all $k$, and therefore the complementary channel for a $1 \rightarrow k$ universal qudit cloner is entanglement-breaking for all $k$. Since the full qudit Unruh channel is a weighted direct sum of $1 \rightarrow k$ universal qudit cloners, we can conclude that the full Unruh channel is a Hadamard channel since each of the elements of the direct sum is a Hadamard channel.
\end{proof}

% Section: Calculating the capacity regions
\section{Capacity regions}
\label{sec:capacity_regions}

We now set out to calculate the full triple trade-off region for the qudit Unruh channel. This region is the set of achievable rates for the resources of classical communication, quantum communication, and quantum entanglement---in a sense, it fully characterizes the capabilities of our channel. A Pareto optimal point of a capacity region is a point such that when considering the trade-off between different resources in the capacity region, improving one resource comes at the expense of another. The quantum dynamic capacity formula enables the calculation of all Pareto optimal points of the various capacity regions of the qudit Unruh channel because it is additive for the case of Hadamard channels~\cite{Wilde_QDC}.
\begin{definition}
The quantum dynamic capacity formula for a quantum Hadamard channel $\mathcal{N}$ is as follows:
\begin{align}
\label{eq:dynamiccapacity}
\mathcal{D}_{\lambda, \mu} (\mathcal{N} ) \equiv \max_{\sigma} \left[ I(AX;B)_{\sigma} + \lambda I(A \rangle BX)_{\sigma} + \mu ( I(X;B)_{\sigma} + I(A \rangle BX)_{\sigma} ) \right]
\end{align}
where $\lambda, \mu \ge 0$, the states $\sigma$ have the form
\begin{align}
\label{eq:classicalquantumstate}
\sigma^{XAB} = \sum_{x \in X} p(x) \ketbra{x}{x}^X \otimes \mathcal{N}^{A' \rightarrow B} ( \phi_x^{AA'} ),
\end{align}
and the states $\phi_x^{AA'}$ are pure.
\end{definition}

\subsection{CQE capacity region}

In order to facilitate our calculation of the state that maximizes~\eqref{eq:dynamiccapacity} for the qudit Unruh channel, we shall exploit the block diagonal structure of the channel and the fact that each block is a universal qudit cloning channel, as shown in Lemma~\ref{lem:blockcloner}. We shall find a state of the form~\eqref{eq:classicalquantumstate} that maximizes~\eqref{eq:dynamiccapacity} for all $\lambda, \mu \ge 0$ for the $1\rightarrow k$ universal qudit cloners that make up the qudit Unruh channel. The states that lead to this maximization will also form the Pareto-optimal points for the full triple trade-off capacity region. The resulting states are given by the following theorem.

% Theorem about states to consider
\begin{theorem}
\label{thm:trade-off-states}
An ensemble of the following form is sufficient to obtain all Pareto-optimal points of the CQE capacity region of a $1 \rightarrow k$ universal qudit cloner:
\begin{align}
\label{eq:CQEoptimalstate}
\dfrac{1}{d} \left( \ketbra{1}{1}^X \otimes \psi_1^{AA'} + \cdots + \ketbra{d}{d}^X \otimes \psi_d^{AA'} \right),
\end{align}
where
\begin{align}
\Tr_A (\psi_1^{AA'})&= (1-\mu_1- \cdots - \mu_{d-1}) \ketbra{1}{1}^{A'} + \mu_1 \ketbra{2}{2}^{A'}  + \cdots + \mu_{d-1} \ketbra{d}{d}^{A'} , \nonumber \\
\Tr_A (\psi_2^{AA'})&= \mu_{d-1} \ketbra{1}{1}^{A'}  + (1-\mu_1- \cdots - \mu_{d-1}) \ketbra{2}{2}^{A'}  + \cdots + \mu_{d-2} \ketbra{d}{d}^{A'} ,\nonumber  \\
\nonumber
&\vdots \\
\Tr_A (\psi_d^{AA'})&= \mu_1 \ketbra{1}{1}^{A'}  + \mu_2 \ketbra{2}{2}^{A'}  + \cdots +  (1-\mu_1- \cdots - \mu_{d-1}) \ketbra{d}{d}^{A'} , \nonumber
\end{align}
and the states $\psi_i$ are pure.

\end{theorem}

% Proof
\begin{proof}
Consider the following classical-quantum state:
\begin{align}
\rho^{XA'}=\sum_x p_X(x) \ketbra{x}{x}^X \otimes \rho_x^{A'}
\end{align}

We shall denote $Cl_{1\rightarrow k}^{(d)}$ to be the $1 \rightarrow k$ universal qudit quantum cloner. We now introduce an augmented classical-quantum state:
\begin{equation}
\sigma^{XJKA'}=\sum_x \sum_{j=0}^{d-1} \sum_{k=0}^{d-1} \dfrac{1}{d^2}  p_X(x) \left( \ketbra{x}{x}^X \otimes \ketbra{j}{j}^J \otimes \ketbra{k}{k}^K \otimes X(j) Z(k) \rho_x^{A'} Z(k)^{\dagger} X(j)^{\dagger} \right)
\end{equation}
where $X(j)$ and $Z(k)$ are the generalized Pauli operators in $d$ dimensions.

The following set of equalities shows that $\sigma^B$ is equal to the maximally mixed state on the symmetric subspace of Bob:
\begin{align}
\label{eq:sigmab}
\sigma^B = Cl_{1\rightarrow k}^{(d)}(\sigma^{A'})=Cl_{1\rightarrow k}^{(d)}\left(\dfrac{I^{A'}}{d}\right) = Cl_{1\rightarrow k}^{(d)} \left(\int V \omega V^{\dagger} dV \right) &= \int R_V Cl_{1\rightarrow k}^{(d)}(\omega) R_{V^{\dagger}}dV \nonumber  \\
&= \dfrac{1}{{k+d-1 \choose d-1}} \sum_{i=1}^{{k+d-1 \choose d-1}} \ketbra{i}{i}
\end{align}
The second equality follows from the fact that an equally weighted mixture of all $d^2$ generalized Pauli matrices produces the maximally mixed state. The fourth equality uses the linearity and covariance of the universal cloning channel.

We now analyze the quantum dynamic capacity formula. Consider the following chain of inequalities:
\begin{align}
\label{eq:QDC}
\nonumber
&I(AX;B)_{\rho} + \lambda I(A \rangle BX)_{\rho} + \mu \left( I(X;B)_{\rho} + I(A \rangle BX)_{\rho} \right) \\
\nonumber
&=H(A|X)_{\rho} + (\mu +1)H(B)_{\rho} + \lambda H(B|X)_{\rho} - (\lambda + \mu + 1) H(E|X)_{\rho} \\
\nonumber
&=H(A|XJK)_{\sigma} + (\mu +1)H(B)_{\rho} + \lambda H(B|XJK)_{\sigma} - (\lambda + \mu + 1) H(E|XJK)_{\sigma} \\
\nonumber
& \le (\mu +1) H(B)_{\sigma} + H(A|XJK)_{\sigma} + \lambda H(B|XJK)_{\sigma} - (\lambda + \mu + 1) H(E|XJK)_{\sigma} \\
\nonumber
&=  (\mu +1) \log{k+d-1 \choose d-1} + H(A|XJK)_{\sigma} + \lambda H(B|XJK)_{\sigma} - (\lambda + \mu + 1) H(E|XJK)_{\sigma} \\
\nonumber
&= (\mu +1) \log{k+d-1 \choose d-1} + \sum_x p_X(x) \left[ H(A)_{\rho_x} + \lambda H(B)_{\rho_x} - (\lambda + \mu + 1) H(E)_{\rho_x} \right] \\
\nonumber
& \le  (\mu +1) \log{k+d-1 \choose d-1} +\max_x \left[ H(A)_{\rho_x} + \lambda H(B)_{\rho_x} - (\lambda + \mu + 1) H(E)_{\rho_x} \right] \\
&= (\mu +1) \log{k+d-1 \choose d-1} +H(A)_{\rho_x^*} + \lambda H(B)_{\rho_x^*} - (\lambda + \mu + 1) H(E)_{\rho_x^*}
\end{align}
The first equality follows from the definition of the quantum mutual information and coherent information. The second equality follows from the fact that the conditional entropies are invariant under the unitary transformations $X(j)$ and $Z(k)$. The first inequality follows from the concavity of entropy and the fact that $\sigma^B$ is a mixture of states of the form of $\rho^B$. The third equality follows from~\eqref{eq:sigmab}.  The fourth equality follows from the fact that the $X$, $J$ and $K$ subsystems are classical. The second inequality follows because $p_X(x)$ is a probability distribution---the weighted sum over the probability distribution is smaller than the maximal value of the term in the square brackets. The final equality follows from defining $\rho_x^*$ to be the state that maximizes the quantity in the square brackets.

The entropies $H(A)_{\rho_x^*}$, $H(B)_{\rho_x^*}$ and $H(E)_{\rho_x^*}$ depend only on the eigenvalues of the input state $\rho^*$ by the covariance of the channel and its complement. Thus, without loss of generality, we can take the state $\rho^*$ to be diagonal in the $\{ \ket{1}, \cdots, \ket{d} \}$ basis of $A'$. The ensemble defined to consist of the purifications of $X(j)\rho^*X(j)^{\dagger}$, $1\le j \le d$, assigned with equal probability for all $j$ then saturates the upper bound in~\eqref{eq:QDC}.
\end{proof}

Recall the classically-enhanced quantum capacity region (CQ) and the entanglement-assisted classical capacity region (CE) outlined in Section~\ref{sec:capacity_regions}. We can now use the state provided in Theorem~\ref{thm:trade-off-states} to calculate the CQ and CE boundary trade-off curves for $1\rightarrow k$ universal cloning channels.

% Theorem about trade-off curves
\begin{theorem}
\label{thm:CQCE_cloners}
The CQ trade-off curve that bounds the classically-enhanced quantum capacity region for a $1 \rightarrow k$ universal qudit cloner is given by the convex hull over the following set of points:
\begin{align}
\left( \log{{k+d-1 \choose d-1}} -  H\left( \dfrac{\alpha(b_1,\cdots, b_d)}{{k+d-1 \choose d}} \right), H \left( \dfrac{\alpha(b_1,\cdots, b_d)}{{k+d-1 \choose d}} \right) -  H \left( \dfrac{\gamma(b_1,\cdots, b_d)}{{k+d-1 \choose d}} \right)  \right).
\label{eq:CE_cloners}
\end{align}
The CE trade-off curve that bounds the entanglement-assisted classical capacity region for a $1 \rightarrow k$ universal qudit cloner is given by the convex hull over the following points:
\begin{align}
\left( \log{{k+d-1 \choose d-1}} -  \sum_{i=1}^d \mu_i \log{\mu_i} - H \left( \dfrac{\gamma(b_1,\cdots, b_d)}{{k+d-1 \choose d}} \right) ,  -  \sum_{i=1}^d \mu_i \log{\mu_i}  \right),
\end{align}
where
\begin{align}
\alpha(b_1 , \cdots, b_d) &= \sum_{i=1}^d \mu_i b_i \qquad \text{such that} \qquad \sum_{i=1}^d b_i = k,  \nonumber \\
\gamma(b_1 , \cdots, b_d) &= 1+ \sum_{i=1}^d \mu_i b_i \qquad \text{such that} \qquad \sum_{i=1}^d b_i = k-1. \nonumber
\end{align}

\end{theorem}

% Proof
\begin{proof}
Now in order to obtain the CQ and CE trade-off curves we consider an ensemble of the form in Theorem~\ref{thm:trade-off-states}. The ensemble we shall consider will thus have the following form:
\begin{align}
\label{eq:CQCE_States}
\dfrac{1}{d} \left( \ketbra{1}{1}^X \otimes \psi_1^{AA'} + \cdots + \ketbra{d}{d} \otimes \psi_d^{AA'} \right),
\end{align}
where
\begin{align}
\Tr_A (\psi_1^{AA'})&= (1-\mu_1- \cdots - \mu_{d-1}) \ketbra{1}{1}^{A'} + \mu_1 \ketbra{2}{2}^{A'}  + \cdots + \mu_{d-1} \ketbra{d}{d}^{A'} , \nonumber \\
\Tr_A (\psi_2^{AA'})&= \mu_{d-1} \ketbra{1}{1}^{A'}  + (1-\mu_1- \cdots - \mu_{d-1}) \ketbra{2}{2}^{A'}  + \cdots + \mu_{d-2} \ketbra{d}{d}^{A'} ,\nonumber  \\
\nonumber
&\vdots \\
\Tr_A (\psi_d^{AA'})&= \mu_1 \ketbra{1}{1}^{A'}  + \mu_2 \ketbra{2}{2}^{A'}  + \cdots +  (1-\mu_1- \cdots - \mu_{d-1}) \ketbra{d}{d}^{A'} . \nonumber
\end{align}
An isometric extension of the $1 \rightarrow k$ universal cloning channel $\mathcal{U}_N$ acts as follows on the above states:
\begin{align}
\ket{\psi_n}^{ABE} = \dfrac{1}{\sqrt{{k+d-1 \choose d}}} \sum_{i=1}^d \sum_{B(k-1)} \sqrt{\mu_{i-n\text{ mod }d}} \sqrt{1+b_i} \ket{i}^A \ket{(b_1,\cdots,1+b_i,\cdots, b_d)}^B \ket{(b_1,\cdots, b_i, \cdots, b_d)} ^E,
\end{align}
where $B(m) = \{ b_i, \text{ }1\le i \le d \text{ } | \sum_{i=1}^d b_i =m \}$.

The output state of the isometric extension is then
\begin{align}
\rho^{XABE}=\dfrac{1}{d} \left( \ketbra{1}{1}^X \otimes \psi_1^{ABE} + \cdots + \ketbra{d}{d}^X \otimes \psi_d^{ABE} \right).
\end{align}

Thus,
\begin{align}
\rho^{XA}&= \dfrac{1}{d} \left(\ketbra{1}{1}^X \otimes \psi_1^A + \cdots + \ketbra{d}{d}^X \otimes \psi_d^A \right), \\
\rho^{XB} &= \dfrac{1}{d} \left(\ketbra{1}{1}^X \otimes \psi_1^B + \cdots + \ketbra{d}{d}^X \otimes \psi_d^B \right), \\
\label{eq:rhob}
\rho^B &= \dfrac{1}{{k+d-1 \choose d-1}} \sum_{B(k)} \ketbra{(b_1,\cdots, b_d)}{(b_1,\cdots,b_d)}, \\
\psi_i^B &= \dfrac{1}{{k+d-1 \choose d}} \sum_{B(k)} (\mu_{1-i \text{ mod }d}b_1+\cdots+\mu_{d-i\text{ mod }d}b_d) \ketbra{(b_1,\cdots, b_d)}{(b_1,\cdots, b_d)}^B.
\end{align}
To obtain~\eqref{eq:rhob} we used the calculation in~\eqref{eq:sigmab}. We can then use these output states to calculate the following entropies:
\begin{align}
H(A|X)&=H(\psi_d^A)=H(\sum_{i=1}^d \mu_i \ketbra{i}{i}) =  - \sum_{i=1}^d \mu_i \log{\mu_i},\\
H(B) &= \log{{k+d-1 \choose d-1}}, \\
\label{eq:HBX}
H(B|X) &= H(\psi_d^B) = H \left( \dfrac{1}{{k+d-1 \choose d}} \sum_{B(k)} (\mu_1b_1+\cdots+\mu_d b_d) \ketbra{(b_1,\cdots, b_d)}{(b_1,\cdots, b_d)} \right),
\end{align}
where we have used the covariance of the channel in~\eqref{eq:HBX}. Thus the Holevo information is given by the following expression:
\begin{align}
I(X;B) &= H(B) - H(B|X) \\
&=\log{{k+d-1 \choose d-1}} -  H\left( \dfrac{1}{{k+d-1 \choose d}} \sum_{B(k)} (\mu_1b_1+\cdots+\mu_d b_d) \ketbra{(b_1,\cdots, b_d)}{(b_1,\cdots, b_d)} \right).
\end{align}
Now to calculate the coherent information $I(A \rangle BX)$ the following states are important:
\begin{align}
\rho^{XE} &= \dfrac{1}{d} \left(\ketbra{1}{1}^X \otimes \psi_1^E + \cdots + \ketbra{d}{d}^X \otimes \psi_d^E \right), \\
\psi_i^E &= \dfrac{1}{{k+d-1 \choose d}} \sum_{B(k-1)} (\mu_{1-i\text{ mod }d}b_1+\cdots+\mu_{d-i\text{ mod }d}b_d+1) \ketbra{(b_1,\cdots, b_d)}{(b_1,\cdots, b_d)}^E .
\end{align}
The coherent information is then given by the following:
\begin{align}
I(A \rangle BX) &= H(B|X) - H(E|X)\\
\nonumber
&= H( \dfrac{1}{{k+d-1 \choose d}} \sum_{B(k)} (\mu_1b_1+\cdots+\mu_d b_d) \ketbra{(b_1,\cdots, b_d)}{(b_1,\cdots, b_d)} ) \\
&\qquad \qquad - H( \dfrac{1}{{k+d-1 \choose d}} \sum_{B(k-1)} (\mu_1b_1+\cdots+\mu_d b_d+1) \ketbra{(b_1,\cdots, b_d)}{(b_1,\cdots, b_d)} ).
\end{align}
\end{proof}

% PDF FIgure
%\begin{figure}[h]
%\centering
%$
%\subfigure[\text{ }Classically-enhanced quantum capacity region]{
%	\includegraphics[width=0.45\textwidth]{CQ_d25crop.pdf}
%	\hspace{.1 in}
%	\label{fig:CQ_cloner}
%	}
%\subfigure[\text{ }Entanglement-assisted classical capacity region]{
%	\includegraphics[width=0.45\textwidth]{CE_d25crop.pdf}
%	\hspace{.1 in}
%	\label{fig:CE_cloner}
%}
%$
%\caption{(Color) Plots of the (a) CQ and (b) CE capacity regions for a $1 \rightarrow k$ universal qudit cloner with $k=1,2,5,10$. The logarithms taken in the calculation of the entropic quantities are in base $d$. Notice that for both dimensions 2 and 5 the classically-enchanced quantum capacity region is convex and larger than that of a time-sharing protocol. Similarly, the entanglement-assisted classical capacity region is convex and is larger than the time-sharing protocol for both dimension 2 and 5. The relative difference between trade-off coding and time-sharing increases as $k$ increases. }
%\label{fig:CQCE_cloner}
%\end{figure}

%% EPS Figure
\begin{figure}
\centering
$
\subfigure[\text{ }Classically-enhanced quantum capacity region]{
	\epsfig{file=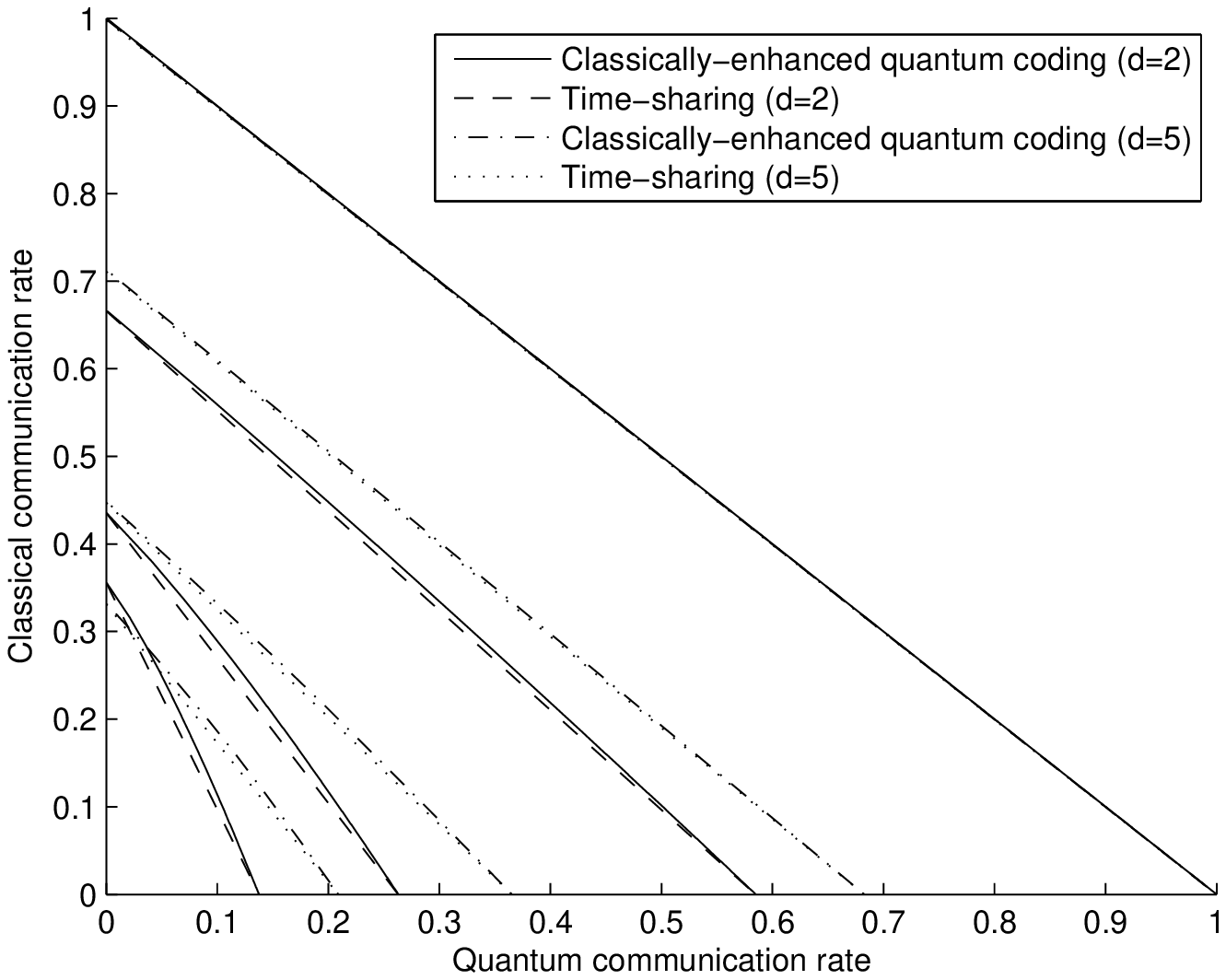, width=0.45\textwidth}
	\hspace{.1 in}
	\label{fig:CQ_cloner}
	}
\subfigure[\text{ }Entanglement-assisted classical capacity region]{
	\epsfig{file=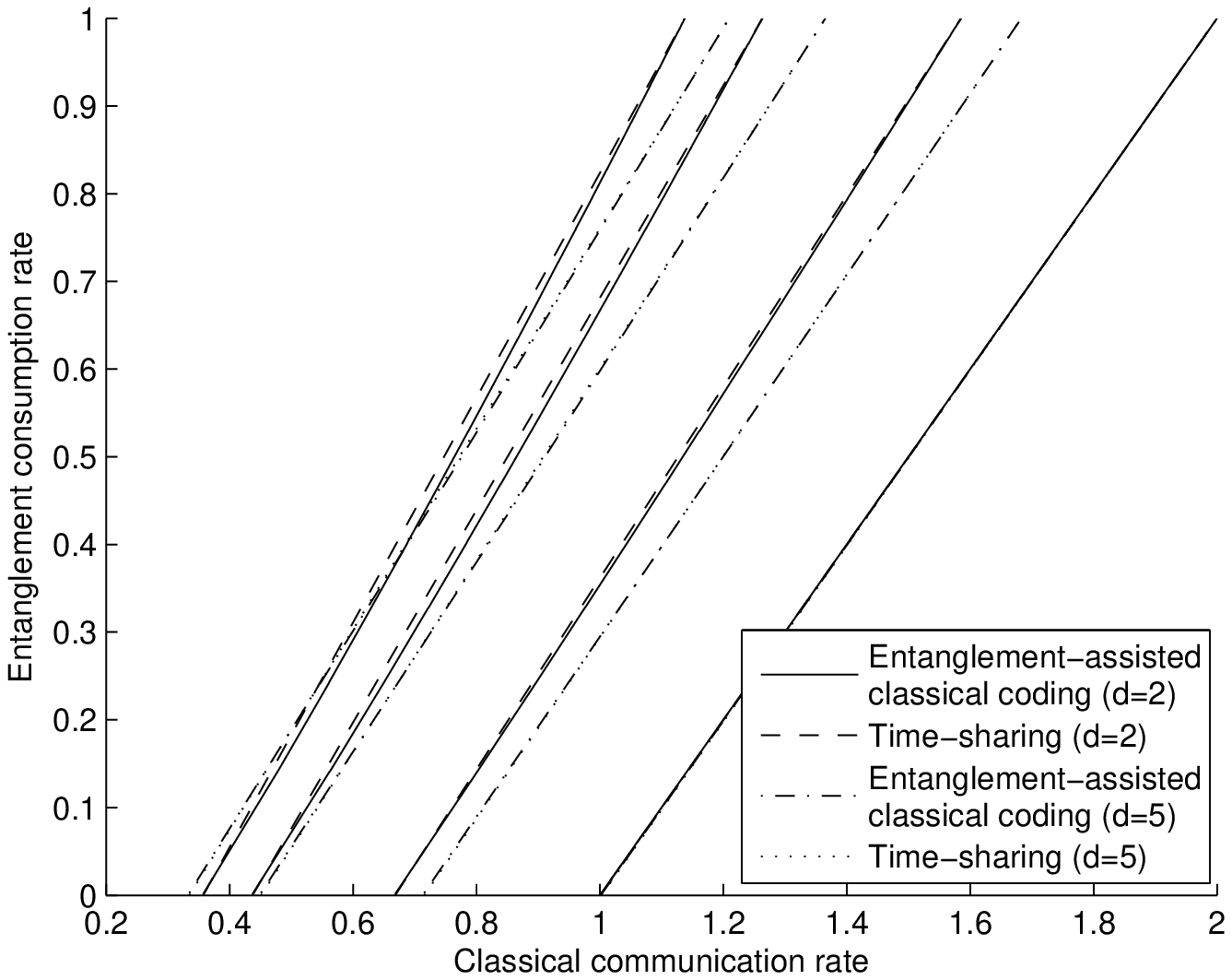, width=0.45\textwidth}
	\hspace{.1 in}
	\label{fig:CE_cloner}
}
$
\caption{Plots of the (a) CQ and (b) CE capacity regions for a $1 \rightarrow k$ universal qudit cloner with $k=1,2,5,10$. The logarithms taken in the calculation of the entropic quantities are in base $d$. Notice that for both dimensions 2 and 5 the classically-enchanced quantum capacity region is convex and larger than that of a time-sharing protocol. Similarly, the entanglement-assisted classical capacity region is convex and is larger than the time-sharing protocol for both dimension 2 and 5. The relative difference between trade-off coding and time-sharing increases as $k$ increases. }
\label{fig:CQCE_cloner}
\end{figure}

Figure~\ref{fig:CQCE_cloner} shows the resulting CQ and CE plots for $1\rightarrow k$ universal qudit cloners for $k=1,2,5,10$ of dimension 5, and compares these to the previously calculated CQ and CE trade-off curves for qubit (dimension 2) universal cloning channels~\cite{Bradler_Hadamard}. In the calculation of the information-theoretic quantities, for example for the calculation of the entropy, we have used a logarithm base~$d$ rather than a logarithm of base~2. We take this convention in order to treat all dimensions on an equal footing. For example, the maximally mixed state always has entropy 1 when one takes the logarithm in base $d$, for any dimension $d$, as opposed to always taking the logarithm in base 2. We produced the plots in Figure~\ref{fig:CQCE_cloner} by considering all possible distributions of input states of the form~\eqref{eq:CQCE_States}. While not all distributions will lead to extremal points, since the region is convex, we can take the convex hull of all these points and produce a region that will capture all extremal points at the boundary. For the CQ trade-off curves, Figure~\ref{fig:CQ_cloner} shows that the classically-enhanced quantum coding scheme achieves better rates than a time-sharing scheme. Time-sharing, in the case of CQ trade-off curves, corresponds to sending classical information a certain fraction of the time and quantum information the rest of the time over the channel. This is characterized on the CQ plot by a straight dashed line (dimension 2) and a straight dotted line (dimension 5) connecting the classical and quantum capacity of the channel in Figure~\ref{fig:CQ_cloner}. Time-sharing, in the CE plot, corresponds to the straight dashed line (dimension 2) and a straight dotted line (dimension 5) connecting the channel's rate of sending purely classical data and classical data assisted by consumption of a dit of entanglement per channel use. Again in the case of entanglement-assisted classical coding, the region is larger than that of time-sharing as characterized by the solid line always being below the dashed line (dimension 2) and the dash dotted line always being below the dotted line (dimension 5) in Figure~\ref{fig:CE_cloner}. Finally, it should be noted that the relative difference between the optimal coding strategy and that of a time-sharing strategy increases as the parameter $k$ increases, as shown by the strong curvature in the classically-enhanced quantum coding and entanglement-assisted coding curves for high values of $k$.

\begin{theorem}
\label{thm:CQEregion}
The quantum dynamic capacity region for the qudit Unruh channel, where $C$ is the classical communication rate, $Q$ is the quantum communication rate, and $E$ is the entanglement generation rate, is given by the convex hull over the regions characterized by the following set of inequalities:
\begin{align}
C + 2Q &\le  -\sum_{i=1}^d \mu_i \log{\mu_i} +\sum_{k=1}^{\infty} p_k(z) ( H(B)_{\rho_k} - H(E|X)_{\rho_k} ), \\
Q+E &\le  \sum_{k=1}^{\infty} p_k(z) (H(B|X)_{\rho_k} - H(E|X)_{\rho_k} ), \\
C+Q+E &\le \sum_{k=1}^{\infty} p_k(z) (H(B)_{\rho_k} - H(E|X)_{\rho_k}),
\end{align}
where $\rho_k$ is the output of the $1 \rightarrow k$ universal qudit cloner with an input of the form of~\eqref{eq:CQEoptimalstate}. A special case of the CQE capacity region is the CQ trade-off curve for the qudit Unruh channel, given by the convex hull over the following set of points:
\begin{align}
\left( \sum_{k=1}^{\infty} p_k(z) ( H(B)_{\rho_k} - H(B|X)_{\rho_k} ) ,  \sum_{k=1}^{\infty}  p_k(z) ( H(B|X)_{\rho_k} - H(E|X)_{\rho_k} ) \right),
\end{align}
and the CE trade-off curve for the qudit Unruh channel, given by the convex hull over the following points:
\begin{align}
\left( \sum_{k=1}^{\infty} p_k(z) ( H(B)_{\rho_k}  - H(E|X)_{\rho_k}) -\sum_{i=1}^d \mu_i \log{\mu_i} , - \sum_{i=1}^d \mu_i \log{\mu_i}  \right).
\end{align}
\end{theorem}

\begin{proof}
We can again use the state from~\eqref{eq:CQEoptimalstate} to obtain the points on the CQE triple trade-off region since the same argument from Theorem~\ref{thm:trade-off-states} can be repeated for the qudit Unruh channel. Thus the output of the isometric extension of the Unruh channel $\mathcal{U}$ for an input state of the form in~\eqref{eq:CQEoptimalstate} is
\begin{align}
&\ket{\psi_n}^{ABE} =\nonumber \\
& \qquad (1-z)^{\frac{d+1}{2}} \bigoplus_{k=1}^{\infty} z^{\frac{k-1}{2}} \sum_{i=1}^d \sum_{B(k-1)} \sqrt{\mu_{i-n\text{ mod }d}} \sqrt{1+b_i} \ket{i}^A \ket{(b_1,\cdots,1+b_i,\cdots, b_d)}^B \ket{(b_1,\cdots, b_i, \cdots, b_d)} ^E.
\end{align}
Now since each $k$ corresponds to a particular block, we can write the output of the Unruh channel as follows:
\begin{align}
\rho^{XABE}&=\dfrac{1}{d} \left( \ketbra{1}{1}^X \otimes \psi_1^{ABE} + \cdots + \ketbra{d}{d}^X \otimes \psi_d^{ABE} \right) \\
& =  (1-z)^{d+1} \bigoplus_{k=1}^{\infty} \dfrac{z^{k-1}}{d} { k+d-1 \choose d } \left( \ketbra{1}{1}^X \otimes \psi_{k,1}^{ABE} + \cdots + \ketbra{d}{d}^X \otimes \psi_{k,d}^{ABE} \right)\\
&=   (1-z)^{d+1} \bigoplus_{k=1}^{\infty} \dfrac{z^{k-1}}{d} { k+d-1 \choose d } \rho_k^{XABE},
\end{align}
where $\psi_{k,i}^{ABE}$ are the outputs for each block, which also correspond to the outputs of a $1 \rightarrow k$ universal qudit cloner, as outlined in the previous theorem. Therefore,
\begin{align}
H(B|X)_{\rho} &= H(\psi_d^B) = H \left( \bigoplus_{k=1}^{\infty} (1-z)^{d+1} z^{ k-1} { k+d-1 \choose d} \psi_{k,d}^B \right)\\
&= \sum_{k=1}^{\infty} H \left(  (1-z)^{d+1} z^{ k-1} { k+d-1 \choose d} \psi_{k,d}^B \right) \\
&= \sum_{k=1}^{\infty} - p_k(z) \log{p_k(z)} + p_k(z) H(B|X)_{\rho_k}, \\
\end{align}
where we have defined $p_k(z) = (1-z)^{d+1} z^{ k-1} { k+d-1 \choose d}$. Thus, similarly we can derive the following:
\begin{align}
H(B)_{\rho} &= \sum_{k=1}^{\infty} - p_k(z) \log{p_k(z)} + p_k(z) H(B)_{\rho_k} ,\\
H(E|X)_{\rho} &= \sum_{k=1}^{\infty} - p_k(z) \log{p_k(z)} + p_k(z) H(E|X)_{\rho_k} .\\
\end{align}
Finally, the quantity $H(A|X)_{\rho}$ remains the same as that of the universal cloning channel, that is, $H(A|X) = - \sum_{i=1}^d \mu_i \log{\mu_i}$.

\end{proof}

% PDF FIgure
%\begin{figure}[H]
%\centering
%%\begin{center}
%$
%\subfigure[\text{ }Classically-enchanced quantum capacity region]{
%	\includegraphics[width=0.45\textwidth]{Unruh_CQ_d25crop.pdf}
%	\hspace{.1 in}
%	\label{fig:CQ_Unruh}
%	}
%\subfigure[\text{ }Entanglement-assisted classical capacity region]{
%	\includegraphics[width=0.45\textwidth]{Unruh_CE_d25crop.pdf}
%	\hspace{.1 in}
%	\label{fig:CE_Unruh}
%}
%$
%\caption{(Color) Plots of the (a) CQ and (b) CE capacity regions for the qudit Unruh channel, for dimensions $d=$ 2, 5, with acceleration parameter $z =$ 0, 0.25, 0.5, 0.75. The logarithms taken in the calculation of the entropic quantities are in base $d$. Both the classically-enhanced quantum region and the entanglement-assisted classical region are larger than that of the time-sharing protocols and the difference between these two regions grows as the acceleration parameter increases (downwards and to the left for the CQ region and upwards and to the left for the CE region).}
%\label{fig:CQCE_Unruh}
%\end{figure}

% EPS FIgure
\begin{figure}
%\centering
\begin{center}
$
\subfigure[\text{ }Classically-enchanced quantum capacity region]{
	\epsfig{file=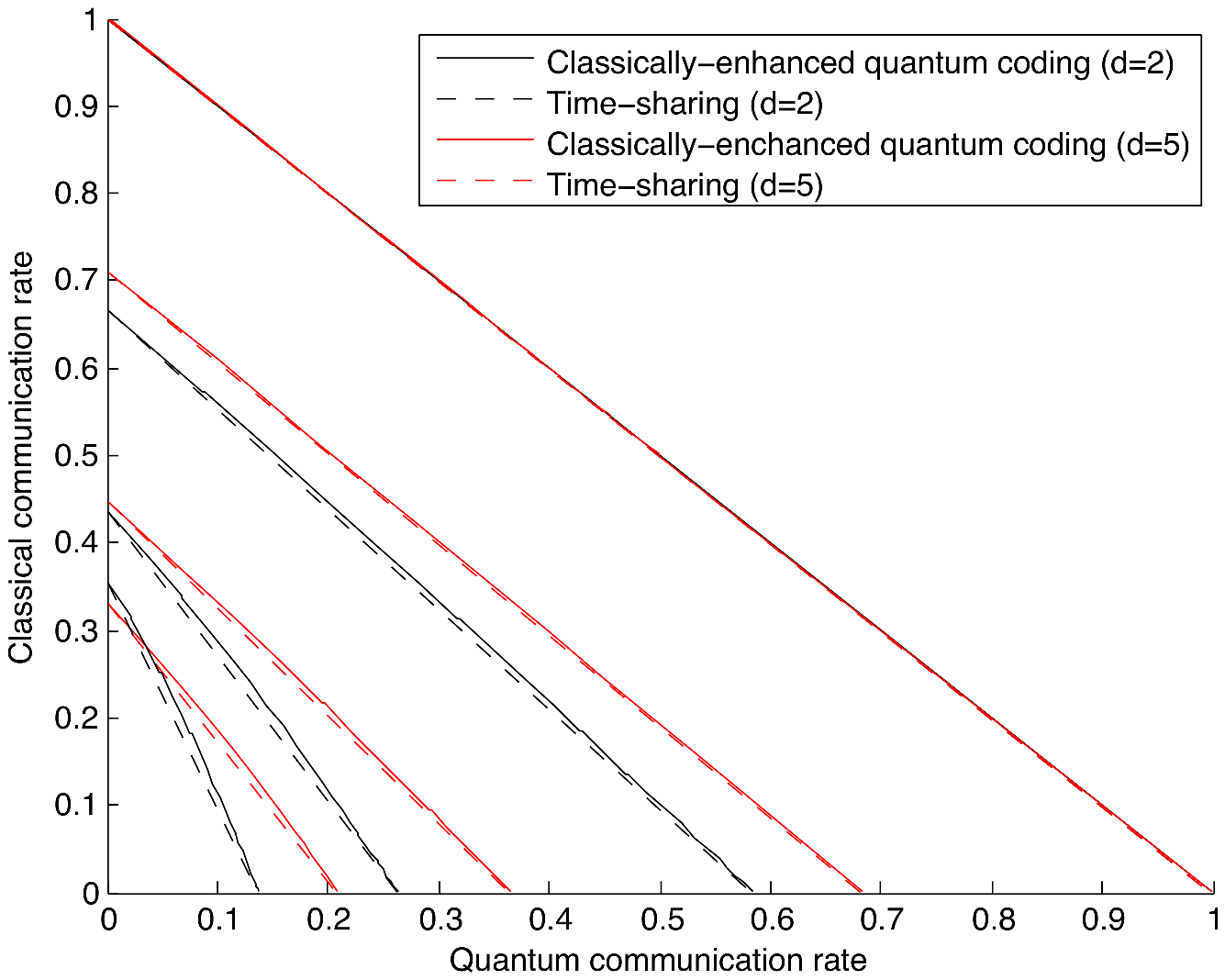, width=0.45\textwidth}
	\hspace{.1 in}
	\label{fig:CQ_Unruh}
	}
\subfigure[\text{ }Entanglement-assisted classical capacity region]{
	\epsfig{file=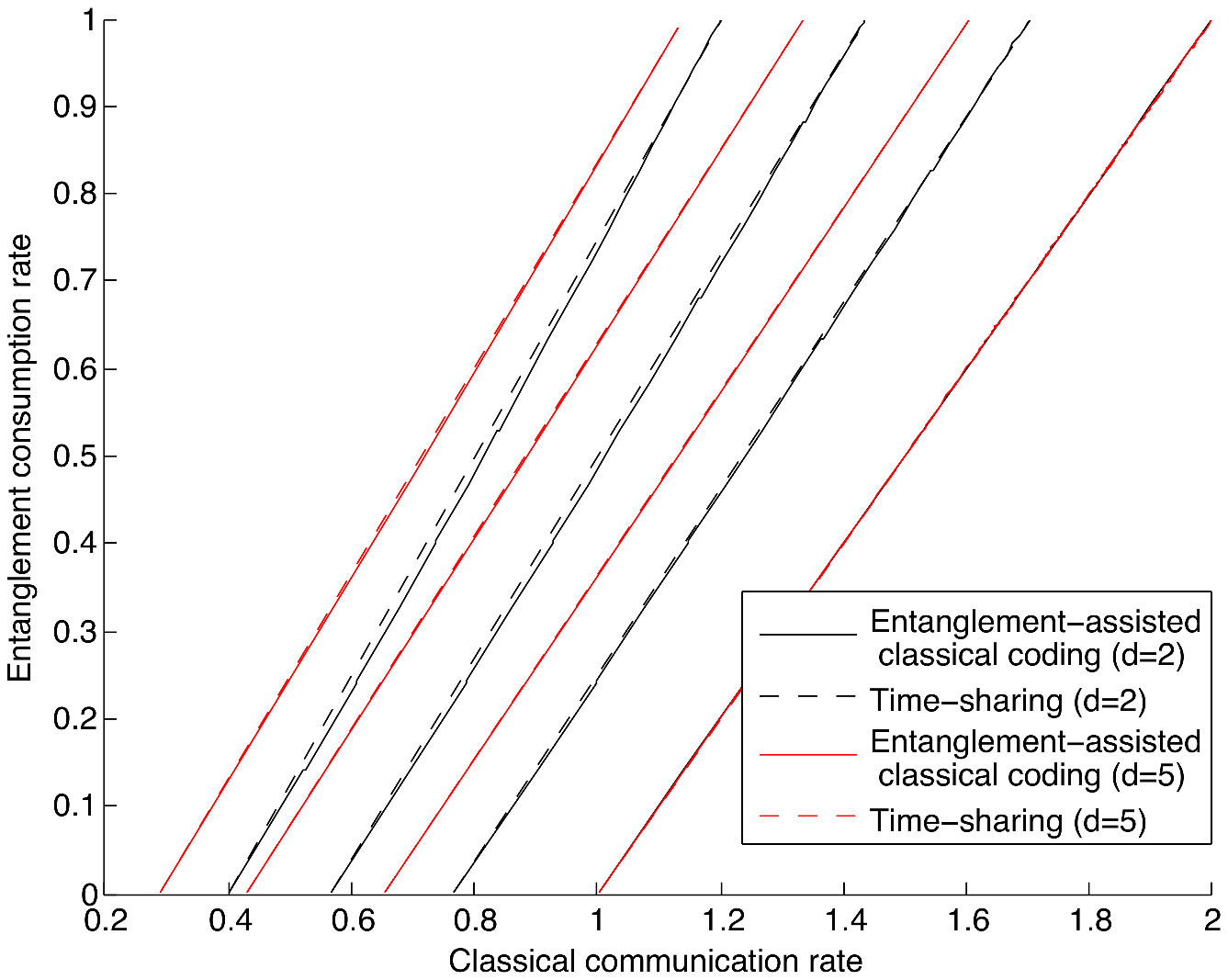, width=0.45\textwidth}
	\hspace{.1 in}
	\label{fig:CE_Unruh}
}
$
\caption{Plots of the (a) CQ and (b) CE capacity regions for the qudit Unruh channel, for dimensions $d=$ 2, 5, with acceleration parameter $z =$ 0, 0.25, 0.5, 0.75. The logarithms taken in the calculation of the entropic quantities are in base $d$. Both the classically-enhanced quantum region and the entanglement-assisted classical region are larger than that of the time-sharing protocols and the difference between these two regions grows as the acceleration parameter increases (downwards and to the left for the CQ region and upwards and to the left for the CE region).}
\label{fig:CQCE_Unruh}
\end{center}
\end{figure}

Figure~\ref{fig:CQCE_Unruh} plots the CQ and CE regions for the qudit Unruh channel for dimensions~2~and~5. The figure shows that both classically-enhanced quantum coding and entanglement-assisted classical coding beat a time-sharing protocol. Moreover, the difference between these optimal coding strategies and the time-sharing strategy increases as the acceleration parameter $z$ increases. This should come as no surprise because a larger value of $z$ implies an increased weight on $1 \rightarrow k$ universal qudit cloners for higher values of $k$. Finally, Figure~\ref{fig:CQE_Unruh} shows the full triple trade-off between the resources of noiseless classical communication, quantum communication, and entanglement consumption for dimension 3 and acceleration parameter~$z=0.75$.

% PDF Figure
%\begin{figure}[H]
%\centering
%\includegraphics[width=0.60\textwidth]{Unruh-CQE-dim-3.pdf}
%\caption{Plot of the full triple-trade off region of the qudit Unruh channel for dimension $d=3$ and acceleration parameter $z=0.75$. The region extends infinitely in several directions as it characterizes the capabilities of the classically-enhanced father protocol~\cite{Hsieh_CQE} combined with the protocols of entanglement distribution (ED), teleportation (TP), and super-dense coding (SD).}
%\label{fig:CQE_Unruh}
%\end{figure}

% EPS Figure
\begin{figure}
\centering
\epsfig{file=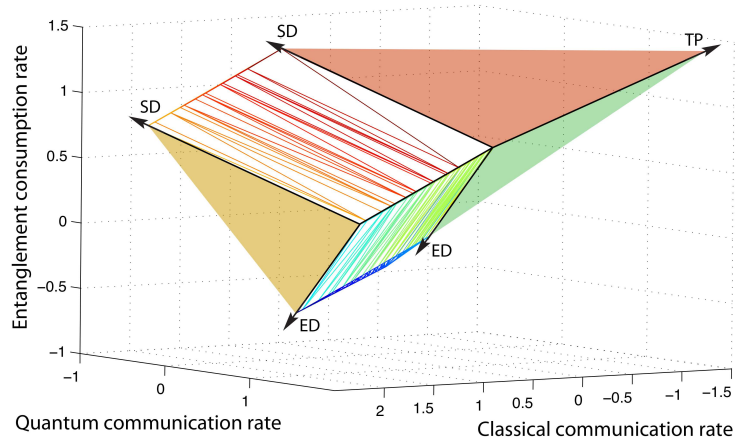, width=0.60\textwidth}
\caption{Plot of the full triple-trade off region of the qudit Unruh channel for dimension $d=3$ and acceleration parameter $z=0.75$. The region extends infinitely in several directions as it characterizes the capabilities of the classically-enhanced father protocol~\cite{Hsieh_CQE} combined with the protocols of entanglement distribution (ED), teleportation (TP), and super-dense coding (SD).}
\label{fig:CQE_Unruh}
\end{figure}

% Subsection about provate quantum capacity
\subsection{Private Dynamic Capacity Formula}

In this section we consider the trade-off between the noiseless classical resources of public communication, private communication, and secret key. Similar to the CQE triple trade-off region, this triple trade-off region characterizes the ability to use these resources along with the quantum channel to generate other noiseless classical resources. As in the case when we were considering quantum resources, we can determine these capabilities with aid from a capacity formula~\cite{Wilde_PrivateQuantumTO}. The private dynamic capacity formula for a quantum Hadamard channel $\mathcal{N}_{H}$ is as follows:
\begin{align*}
P_{\lambda,\mu}\left(  \mathcal{N}_{H}\right)  =\max_{\omega}I\left(
YX;B\right)  _{\omega}+\lambda\left[  H\left(  B|X\right)  _{\omega}-H\left(
E|X\right)  _{\omega}\right]  +\mu\left[  H\left(  B\right)  _{\omega
}-H\left(  E|X\right)  _{\omega}\right]  ,
\end{align*}
where $\omega$ is a state of the following form:%
\begin{align*}
\omega^{XYBE}\equiv\sum_{x,y}p_{X,Y}\left(  x,y\right)  \left\vert
x\right\rangle \left\langle x\right\vert ^{X}\otimes\left\vert y\right\rangle
\left\langle y\right\vert ^{Y}\otimes \mathcal{U}_{\mathcal{N}_{H}}\left(  \psi
_{x,y}\right),
\end{align*}
$\mathcal{U}_{\mathcal{N}_{H}}$ is an isometric extension of the channel
$\mathcal{N}_{H}$, and the states $\psi_{x,y}$ are pure.

\begin{theorem}
\label{thm:RPS_states}
An ensemble of the following form suffices to maximize the private dynamic
capacity formula of a $1\rightarrow k$ universal qudit cloning channel:
\begin{align}
\label{eq:RPS_states}
\dfrac{1}{d} \left( \ketbra{1}{1}^X \otimes \psi_1^{YA'} + \cdots + \ketbra{d}{d}^X \otimes \psi_d^{YA'} \right)
\end{align}
where
\begin{align}
\psi_1^{YA'} &= (1 - \nu_1 - \cdots \nu_{d-1}) \ketbra{1}{1}^Y \otimes \ketbra{1}{1}^{A'} + \nu_1 \ketbra{2}{2}^Y \otimes \ketbra{2}{2}^{A'} + \cdots + \nu_{d-1} \ketbra{d}{d}^Y \otimes \ketbra{d}{d}^{A'}, \nonumber \\
&\vdots \nonumber \\
\psi_d^{YA'} &= \nu_1 \ketbra{1}{1}^Y \otimes \ketbra{1}{1}^{A'} + \nu_2 \ketbra{2}{2}^Y \otimes \ketbra{2}{2}^{A'}  \cdots + (1- \nu_1 - \cdots - \nu_{d-1}) \ketbra{d}{d}^Y \otimes \ketbra{d}{d}^{A'}.  \nonumber
\end{align}
\end{theorem}

\begin{proof}
We shall denote the $1 \rightarrow k$ universal qudit cloner as $Cl_{1\rightarrow k}^{(d)}$ and its complement by $\mathcal{S}_k^{(d)}$. We exploit the following classical-quantum states:%
\begin{align*}
\rho^{XYA^{\prime}} &  \equiv\sum_{x,y}p_{X}\left(  x\right)  p_{Y|X}\left(
y|x\right)  \left\vert x\right\rangle \left\langle x\right\vert ^{X}%
\otimes\left\vert y\right\rangle \left\langle y\right\vert ^{Y}\otimes
\phi_{x,y}^{A^{\prime}},\\
\sigma^{XYIJA^{\prime}} &  \equiv\sum_{x,y}\sum_{i,j=0}^{d-1}\frac{1}{d^{2}%
}p_{X}\left(  x\right)  p_{Y|X}\left(  y|x\right)  \left\vert x\right\rangle
\left\langle x\right\vert ^{X}\otimes\left\vert y\right\rangle \left\langle
y\right\vert ^{Y}\otimes\left\vert i\right\rangle \left\langle i\right\vert
^{I}\otimes\left\vert j\right\rangle \left\langle j\right\vert ^{J}\otimes
X\left(  i\right)  Z\left(  j\right)  \phi_{x,y}^{A^{\prime}}Z^{\dag}\left(
j\right)  X^{\dag}\left(  i\right)  ,
\end{align*}
where the states $\phi_{x,y}^{A^{\prime}}$ are pure, $X\left(  i\right)  $ and
$Z\left(  j\right)  $ are the generalized Pauli operators and let $\rho
^{XYBE}$ and $\sigma^{XYIJBE}$ be the states obtained by transmitting the
$A^{\prime}$ system through the isometric extension of the universal cloning channel.
Let $\sigma_{x}^{A^{\prime}Y}\equiv\sum_{y}p_{Y|X}\left(  y|x\right)
\left\vert y\right\rangle \left\langle y\right\vert ^{Y}\otimes\phi
_{x,y}^{A^{\prime}}$.

The universal qudit cloning channel is covariant by construction~\cite{Gisin, Werner, Keyl} and
the following relationships hold for any input density operator $\sigma$ and
any unitary $V$ acting on the input system $A^{\prime}$:%
\begin{align*}
Cl_{1\rightarrow k}^{(d)}\left(  V\sigma V^{\dag}\right)   &  =R_{V}%
Cl_{1\rightarrow k}^{(d)}\left(  \sigma\right)  R_{V}^{\dag},\\
\mathcal{S}_k^{(d)}\left(  V\sigma V^{\dag}\right)   &  =S_{V}%
\mathcal{S}_k^{(d)}\left(  \sigma\right)  S_{V}^{\dag},
\end{align*}
where $R_{V}$ and $S_{V}$ are higher-dimensional irreducible representations
of the unitary $V$ on the respective systems $B$ and $E$. The state
$\sigma^{B}$ is equal to the maximally mixed state on the completely symmetric
subspace of $N$ qudits for the following reasons:%
\begin{align}
\sigma^{B}=Cl_{1\rightarrow k}^{(d)}\left(  \sigma^{A^{\prime}}\right)
=Cl_{1\rightarrow k}^{(d)}\left(  \frac{I^{A^{\prime}}}{2}\right)
=Cl_{1\rightarrow k}^{(d)}\left(  \int V\omega V^{\dag}\ \text{d}V\right)  &=\int
R_{V}Cl_{1\rightarrow k}^{(d)}\left(  \omega\right)  R_{V^{\dag}}\ \text{d}V \nonumber \\
&=\frac{1}{p_{k}^{d}}\sum_{i=0}^{p_{k}^{d}-1}\left\vert i\right\rangle \left\langle
i\right\vert ^{B},\label{eq:cloning-unital-relation}%
\end{align}
where $p_{k}^{d}\equiv\binom{k+d-1}{d-1}$ and the fourth equality exploits the
linearity and covariance of the universal qudit cloning channel $Cl_{1\rightarrow k}^{(d)}$.

Consider the following chain of inequalities:%
\begin{align*}
&  H\left(  B\right)  _{\rho}-H\left(  B|YX\right)  _{\rho}+\lambda\left[
H\left(  B|X\right)  _{\rho}-H\left(  E|X\right)  _{\rho}\right]  +\mu\left[
H\left(  B\right)  _{\rho}-H\left(  E|X\right)  _{\rho}\right]  \\
&  =\left(  \mu+1\right)  H\left(  B\right)  _{\rho}-H\left(  B|YX\right)
_{\rho}+\lambda H\left(  B|X\right)  _{\rho}-\left(  \lambda+\mu\right)
H\left(  E|X\right)  _{\rho}\\
&  =\left(  \mu+1\right)  H\left(  B\right)  _{\rho}-H\left(  B|YXIJ\right)
_{\sigma}+\lambda H\left(  B|XIJ\right)  _{\sigma}-\left(  \lambda+\mu\right)
H\left(  E|XIJ\right)  _{\sigma}\\
&  \leq\left(  \mu+1\right)  H\left(  B\right)  _{\sigma}-H\left(
B|YXIJ\right)  _{\sigma}+\lambda H\left(  B|XIJ\right)  _{\sigma}-\left(
\lambda+\mu\right)  H\left(  E|XIJ\right)  _{\sigma}\\
&  =\left(  \mu+1\right)  \log\left(  p_{k}^{d}\right)  -\sum_{x,y}%
p_{X}\left(  x\right)  p_{Y|X}\left(  y|x\right)  H(B)_{Cl_{1\rightarrow k}^{(d)}(\phi)}  \\
& \qquad \qquad \qquad \qquad \qquad +\sum_{x}p_{X}\left(  x\right)  \left[  \lambda H\left(
B\right)  _{Cl_{1\rightarrow k}^{(d)}(\sigma_{x}^{A^{\prime}})}-\left(  \lambda+\mu\right)
H\left(  E\right)  _{\mathcal{S}_k^{(d)}(\sigma_{x}^{A^{\prime}})}\right]  \\
&  \leq\left(  \mu+1\right)  \log\left(  p_{k}^{d}\right)  -H(B)_{Cl_{1\rightarrow k}^{(d)}(\phi)}  +\lambda H\left(  B\right)  _{Cl_{1\rightarrow k}^{(d)}(\sigma
_{x}^{\ast})}-\left(  \lambda+\mu\right)  H\left(  E\right)  _{\mathcal{S}_k^{(d)}(\sigma_{x}^{\ast})}\\
&  =\log\left(  p_{k}^{d}\right)  -H(B)_{Cl_{1\rightarrow k}^{(d)}(\phi)}
+\lambda\left[  H\left(  B\right)  _{Cl_{1\rightarrow k}^{(d)}(\sigma_{x}^{\ast})}-H\left(
E\right)  _{\mathcal{S}_k^{(d)}(\sigma_{x}^{\ast})}\right]  + \mu \left[  \log\left(
p_{k}^{d}\right)  -H\left(  E\right)  _{\mathcal{S}_k^{(d)}(\sigma_{x}^{\ast}%
)}\right]  .
\end{align*}
The first equality follows by rearranging terms. The second equality follows
because the conditional entropies are invariant under unitary transformations:%
\begin{align*}
H(B)_{R_{\sigma_{j}}\rho_{x}^{B}R_{\sigma_{j}}^{\dag}}=H(B)_{\rho_{x}^{B}%
},\ \ \ \ \ \ H(E)_{S_{\sigma_{j}}\rho_{x}^{E}S_{\sigma_{j}}^{\dag}%
}=H(E)_{\rho_{x}^{E}},
\end{align*}
where $R_{\sigma_{j}}$ and $S_{\sigma_{j}}$ are higher-dimensional
representations of $\sigma_{j}$ on systems $B$ and $E$, respectively. The
first inequality follows because entropy is concave, i.e., the local state
$\sigma^{B}$ is a mixed version of $\rho^{B}$. The third equality follows
because (\ref{eq:cloning-unital-relation}) implies that $H(B)_{\sigma^{B}%
}=\log\left(  p_{k}^{d}\right)  $, from applying unitary covariance of the
universal cloning channel to the term $H\left(  B|YXI\right)  _{\sigma}=\sum_{x,y}%
p_{X}\left(  x\right)  p_{Y|X}\left(  y|x\right)  H\left(  B\right)
_{Cl_{1\rightarrow k}^{(d)}\left(  \phi_{x,y}\right)  }$, since all pure states have the same
output entropy we have just replaced the state $\phi_{x,y}$ by an arbitrary pure state $\phi$ of our choice, 
and from expanding the conditional entropies
$H\left(  B|XIJ\right)  _{\sigma}$ and $H\left(  E|XIJ\right)  _{\sigma}$. The
second inequality follows because the maximum value of a realization of a
random variable is not less than its expectation. The final equality follows
by rearranging terms.

The entropies $H(B)_{Cl_{1\rightarrow k}^{(d)}(\sigma_{x}^{\ast})}$ and $H(E)_{\mathcal{S}_k^{(d)}(\sigma_{x}^{\ast})}$ depend only on the eigenvalues of the input state
$\sigma_{x}^{\ast}$ by the covariance of both the universal cloning channel and its
complement. We can therefore choose $\sigma_{x}^{\ast}$ to be a state diagonal
in the $\{\ket{1}, \ket{2} ,\cdots ,\ket{d} \}$ basis of
$A^{\prime}$, and without loss of generality, suppose these eigenvalues are
equal to $\left\{  \nu_{1},\ldots,\nu_{d}\right\}  $. The ensemble defined
to consist of $\left\{  X\left(  j\right)  \sigma_{x}^{\ast}X^{\dag}\left(
j\right)  \right\}  _{j=1}^{d}$ assigned equal probabilities then saturates
the upper bound.
\end{proof}

\begin{theorem}
The private dynamic capacity region for the qudit Unruh channel is given by the convex hull of the regions characterized by the following set of inequalities, where $R$ is the rate of classical public communication, $P$ is the rate of classical private communication, and $S$ is the rate of secret key generation:
\begin{align}
R +P &\le \sum_{k=1}^{\infty} p_k(z) \left( H(B)_{\rho_k} - \sum_{b=1}^k  m_{b,k} \dfrac{b}{M_k} \log{\dfrac{b}{M_k}} \right), \\
P+S &\le \sum_{k=1}^{\infty} p_k(z) \left( H(B|X)_{\rho_k} -  H(E|X)_{\rho_k} \right), \\
R+P+S &\le \sum_{k=1}^{\infty} p_k(z) \left( H(B)_{\rho_k} - H(E|X)_{\rho_k} \right) .
\end{align}
$\rho_k$ is the output state of a $1 \rightarrow k$ universal qudit cloner with an input state of the form~\eqref{eq:CQEoptimalstate}, $m_b = { (k-b) + d -2 \choose d-2}$, and $M_k = { k +d-1 \choose d}$.
\end{theorem}

\begin{proof}
Theorem~\ref{thm:RPS_states} provided the form of the input states needed to calculate the private dynamic capacity region for universal qudit cloning channels and since the qudit Unruh channel is composed of such channels, these states will be sufficient to characterize the private dynamic capacity region for the qudit Unruh channel. We now proceed to calculating the various entropies required to characterize the region for the output state $\rho^{XYBE}$.

\begin{align}
\rho^{XYBE} &= \dfrac{1}{d} (\ketbra{1}{1}^X \otimes \psi_1^{YBE} + \cdots + \ketbra{d}{d}^X \otimes \psi_d^{YBE}) \\
&= (1-z)^{d+1} \bigoplus_{k=1}^{\infty} \dfrac{z^{k-1}}{d} {k+d-1 \choose d} (\ketbra{1}{1}^X \otimes \psi_{k,1}^{YBE} + \cdots + \ketbra{d}{d}^X \otimes \psi_{k,d}^{YBE},
\end{align}
where $\psi_{k,d}^{YBE}$ are the output states of the $1 \rightarrow k$ universal qudit cloner. Since the form of the input state given by~\eqref{eq:RPS_states} is identical to the input state for the quantum dynamic capacity region, given by~\eqref{eq:CQEoptimalstate}, the entropic quantities $H(B)$, $H(B|X)$, and $H(E|X)$ will all be equal to those calculated in Theorem~\ref{thm:CQEregion}. That is,
\begin{align}
H(B) &= \sum_{k=1}^{\infty} -p_k(z) \log{p_k(z)} + p_k(z) H(B)_{\rho_k}, \\
H(B|X) &= \sum_{k=1}^{\infty} -p_k(z) \log{p_k(z)} + p_k(z) H(B|X)_{\rho_k}, \\
H(E|X) &= \sum_{k=1}^{\infty} -p_k(z) \log{p_k(z)} + p_k(z) H(E|X)_{\rho_k}, \\
\end{align}
where $\rho_k$ is the output of the $1 \rightarrow k$ universal qudit cloner with an input state of the form~\eqref{eq:CQEoptimalstate}, as described in Theorem~\ref{thm:CQCE_cloners}. Finally, if one conditions on both the $X$ and $Y$ subsystems, then the output state $\rho_{x,y}^{BE}$ is pure for a given $x$ and $y$. Since the channel is covariant, one can calculate the entropies $H(B|XY)$ and $H(E|XY)$ by only considering an input state of the form $\ketbra{1}{1}$. Thus,
\begin{align}
H(B|XY)_{\rho} &= H( \mathcal{N}(\ketbra{1}{1} ) = H \left( \bigoplus_{k=1}^{\infty} p_k(z) Cl_{1\rightarrow k}^{(d)} (\ketbra{1}{1}) \right) \\
&= \sum_{k=1}^{\infty} -p_k(z) \log{p_k(z)} + p_k(z) H(Cl_{1\rightarrow k}^{(d)} (\ketbra{1}{1}) ) \\
&= \sum_{k=1}^{\infty} -p_k(z) \log{p_k(z)} + p_k(z)  \sum_{b=1}^k  m_{b,k} \dfrac{b}{M_k} \log{\dfrac{b}{M_k}},
\end{align}
where the quantities $m_b = { (k-b) + d -2 \choose d-2}$ and $M_k = { k +d-1 \choose d}$ are calculated in Lemma~\ref{lem:blockcloner}. Since $H(B|XY)$ is calculated by taking the marginal entropy of a pure state, the marginal entropy $H(E|XY)$ must be identical, that is $H(E|XY) = H(B|XY)$.

\end{proof}

% PDF Figure
%\begin{figure}[H]
%\centering
%\includegraphics[width=0.60\textwidth]{Unruh-RPS.pdf}
%\caption{Private dynamic capacity region of the qudit Unruh channel for dimension~$d=5$ and acceleration parameter $z=0.75$. The region combines the publicly-enhanced private father protocol~\cite{Hsieh_PPcommunication} along with the protocols of the one-time pad (OTP), private-to-public communication (P2P), and secret key distribution (SKD).}
%\label{fig:PRS_Unruh}
%\end{figure}

% EPS Figure
\begin{figure}
\centering
\epsfig{file=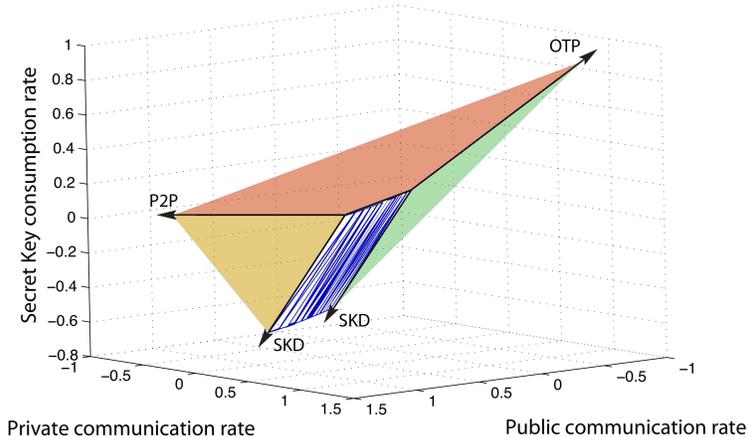, width=0.60\textwidth}
\caption{Private dynamic capacity region of the qudit Unruh channel for dimension~$d=5$ and acceleration parameter $z=0.75$. The region combines the publicly-enhanced private father protocol~\cite{Hsieh_PPcommunication} along with the protocols of the one-time pad (OTP), private-to-public communication (P2P), and secret key distribution (SKD).}
\label{fig:PRS_Unruh}
\end{figure}

We can then proceed to calculate the private dynamic capacity region for the qudit Unruh channel as outlined in Section~\ref{sec:PDC_region}. Figure~\ref{fig:PRS_Unruh} shows the private dynamic capacity region for dimension $d=5$ and acceleration coefficient $z=0.75$. The curvature of one part of the region's boundary demonstrates that the use of the publicly-enhanced private father protocol~\cite{Hsieh_PPcommunication} allows for the sender and receiver to achieve rates that surpass those from a time-sharing strategy.

\section{Conclusion}
\label{sec:conclusion}
We have calculated the full triple trade-off capacity region for the $d$-dimensional Unruh channel by proving that it belongs to the class of Hadamard channels, whose trade-off capacity formulas single-letterize.
% The proof of this result follows by demonstrating that the $d$-dimensional Unruh channel belongs to the class of Hadamard channels.
While this is the same approach used for the qubit (2-dimensional) Unruh channel~\cite{Bradler_Hadamard}, the construction of the proof is different from that of previous works in order to capture the full generality of the dimension $d$. The first feature of our proof is that the Unruh channel has a direct sum form, where the elements belonging to the direct sum are $d$-dimensional universal qudit cloners. Secondly, we have provided an explicit construction of the Kraus operators for the complementary channel of the $1 \rightarrow 2$ universal qudit cloner, enabling us to conclude that the Unruh channel itself is Hadamard through an inductive argument again based on the particular direct sum form of the channel.

The quantum dynamic capacity region for the qudit Unruh channel has a curved boundary, demonstrating that the use of the classically-enchanced father protocol~\cite{Hsieh_CQE} achieves rates superior to a time-sharing strategy.  The relative gain coming from the use of the improved coding strategies increases as the acceleration between the two parties increase. This is due to an increased weight on the $1 \rightarrow k$ universal qudit cloners with higher $k$ when the acceleration between the two parties is large since the probability distribution favours higher values of $k$ as the acceleration parameter $z$ grows. Similarly, the private dynamic capacity region for the qudit Unruh channel also demonstrates that the use of the publicly-enhanced private father protocol~\cite{Hsieh_PPcommunication} achieves rates greater than that of a time-sharing strategy when studying the trade-off of public communication, private communication, and shared secret key.
%Further, the capacity region has an increasing number of faces as the dimension of the Unruh channel grows which is a direct manifestation of the fact that the optimization program one has to consider in the calculation of such capacity regions is over a set of parameters that grows with dimension size.

The connection between the qudit Unruh channel and the universal qudit cloning channels is powerful. Beyond being useful for calculations of entropic quantities, it would be interesting to pursue this connection further in quantum optics. A $1 \rightarrow 2$ universal qubit cloning channel has been experimentally realized using light amplifiers based on the physical process of stimulated emission in an optical fiber~\cite{Fasel}. There have been further proposals for qudits using time-bin entangled photons in optical fibers~\cite{deRiedmatten, Thew} which would then naturally lead to the following question: Could one implement the qudit Unruh channel in terms of these universal optical cloning devices through a network of optical fibers and amplifiers? In such a scheme, the acceleration parameter that was present in the relativistic setting would now be considered as the gain of optical amplifiers characterizing the ``strength'' of the amplification process.

While we have extended the analysis of the capacity regions to encompass the more general $d$-dimensional Unruh channel, it would be interesting to extend this generalization in future works. One could consider the capacity regions of this relativistic setting upon imposing an energy constraint on the system, which would constrain the mean number of photons to which Alice has access in order to encode her information. Would the freedom to encode in any number of modes which satisfy the energy constraint provide an added benefit to that of being restricted to encoding in a fixed $d$-dimensional encoder? If so, would the shape of the capacity regions be severely altered?

\section{Acknowledgements}
T. J. would like to acknowledge the support of NSERC through the USRA Program and the Alexander Graham Bell Canada Graduate Scholarship. K.B. greatly acknowledges the Office of Naval Research (grant No.~N000140811249). M. M. W. acknowledges support from the MDEIE (Qu\'{e}bec) PSR-SIIRI international collaboration grant. The authors would like to acknowledge Norbert~L\"utkenhaus and Patrick~Hayden for useful discussions.

% Bibliography

\appendix

\section{$1 \rightarrow 2$ universal qudit cloning channel is Hadamard}
\label{app:entanglementbreaking}

\begin{proof}[Proof of Lemma~\ref{lem:entanglementbreaking}]
States $\ket{\psi(\vec{\textbf{n}})}$ form a set of rank-one operators. What remains to be shown is the following:
% Requirements
\begin{align}
\mathcal{S}_2^{(d)}(\rho)=\sum_i E_i\rho E_i^{\dagger},
\label{eq:1}
\end{align}
\begin{align}
I=\sum_i E_i^{\dagger}E_i,
\label{eq:2}
\end{align}
where $E_i$ are the rank-one Kraus operators outlined above.

The channel $\mathcal{S}_2^{(d)}$ performs the following mapping on an arbitrary pure input state:
% Form of complementary channel block
\begin{align}
&\mathcal{S}_2^{(d)}(\rho)=\mathcal{S}_2^{(d)}\left(\sum_{i=1}^3\beta_ia_i^{\dagger}\ket{\text{vac}} \right) \\
&=
\dfrac{1}{d+1}
\begin{pmatrix}
2|\beta_1|^2+|\beta_2|^2+\cdots+|\beta_d|^2 & \beta_1\overline{\beta}_2 & \cdots &\beta_1\overline{\beta}_d \\
\overline{\beta}_1\beta_2  & |\beta_1|^2+2|\beta_2|^2+\cdots+|\beta_d|^2 & & \beta_2\overline{\beta}_d \\
\vdots & & \ddots & \vdots\\
\overline{\beta}_1\beta_d & &  \cdots & |\beta_1|^2+\cdots+|\beta_{d-1}|^2+2|\beta_d|^2
\end{pmatrix}.
\label{eq:outputmatrix}
\end{align}
In order to show \eqref{eq:1}, consider the matrix entry $\ketbra{p}{q}$ of $\sum_i E_i\rho E_i^{\dagger}$,
\begin{align}
(\sum_i E_i\rho E_i^{\dagger})_{pq}&=\sum_{n_2,\cdots,n_d=0}^3 \left(\dfrac{1}{4^{d-1}(d+1)}(i^{n_p}\ket{p})(\sum_j i^{n_j} \bra{j})(\sum_{k,l}\beta_k \overline{\beta}_l \ketbra{k}{l})(\sum_{j'}(-i)^{n_{j'}}\ket{j'})((-i)^{n_q}\bra{q})\right) \nonumber \\
&\qquad \qquad+\sum_{k,l}\dfrac{\beta_k\overline{\beta}_l}{d+1}\ket{p}\braket{p}{k}\braket{l}{q} \bra{q}\delta(p-q)\\
&=\sum_{n_2,\cdots,n_d=0}^3 \left(\dfrac{i^{n_p+n_q}(-1)^{n_q}}{4^{d-1}(d+1)}\sum_{j,j',k,l}i^{n_j+n_{j'}}(-1)^{n_{j'}}\beta_k \overline{\beta}_l \braket{j}{k}\braket{l}{j'}\right)\ketbra{p}{q} \nonumber \\
& \qquad \qquad + \dfrac{\beta_p \overline{\beta}_q}{d+1}\delta(p-q)\ketbra{p}{q}\\
&=\sum_{n_2,\cdots,n_d=0}^3 \left(\dfrac{i^{n_p+n_q}(-1)^{n_q}}{4^{d-1}(d+1)}\sum_{k,l} i^{n_k+n_l}(-1)^{n_l}\beta_k \overline{\beta}_l\right)\ketbra{p}{q}+\dfrac{\beta_p \overline{\beta}_q}{d+1}\delta(p-q)\ketbra{p}{q} .
\label{eq:KrausCheck}
\end{align}
Consider first the case when $p=q$, \eqref{eq:KrausCheck} then becomes
\begin{align}
(\sum_iE_i \rho E_i^{\dagger})_{pp}&=\sum_{n_2,\cdots,n_d=0}^3 \left(\dfrac{ i^{2n_p}(-1)^{n_q}}{4^{d-1}(d+1)}\sum_{k,l} i^{n_k+n_l}(-1)^{n_l}\beta_k \overline{\beta}_l\right)\ketbra{p}{p}+\dfrac{|\beta_p|^2}{d+1}\ketbra{p}{p}\\
&=\sum_{n_2,\cdots,n_d=0}^3 \left(\dfrac{(-1)^{2n_p}}{4^{d-1}(d+1)}\sum_{k,l} i^{n_k+n_l}(-1)^{n_l}\beta_k \overline{\beta}_l \right)\ketbra{p}{p}+\dfrac{|\beta_p|^2}{d+1}\ketbra{p}{p}\\
&=\dfrac{1}{4^{d-1}(d+1)}\sum_{n_2,\cdots,n_d=0}^3 \left(\sum_{k,l \atop(k \ne l)} i^{n_k+n_l}(-1)^{n_l}\beta_k \overline{\beta}_l+\sum_k  i^{2n_k}(-1)^{n_k} |\beta_k|^2\right)\ketbra{p}{p} \nonumber \\
& \qquad \qquad +\dfrac{|\beta_p|^2}{d+1}\ketbra{p}{p} \label{eq:CancellationStart}.
\end{align}
One can show the first double sum in~\eqref{eq:CancellationStart} is equal to zero by permuting over all possible values of $n_k$, which are independent of $n_l$. In order to explicitly show the cancellation of these terms, fix values for $k$ and $l$ such that $k \ne l$,
\begin{align}
\sum_{n_2,\cdots,n_d=0}^3 i^{n_k} (-i)^{n_l} \beta_k \overline{\beta}_l = \sum_{n_k=0}^3 i^{n_k} \beta_k \sum_{n_l=0}^3 (-i)^{n_l} \overline{\beta}_l=0
\end{align}
since $\sum_{n_l=0}^3 (-i)^{n_l} \overline{\beta}_l=0$ and this sum is independent of the value of $n_k$ since $k \ne l$. It is worth noting that in the case when $l=1$ this is not true since $n_1$ is fixed to be equal to 1. However in this case we obtain $\overline{\beta}_1\sum_{n_k=0}^3 i^{n_k} \beta_k=0$ since the sum over $k$ is now equal to 0 because $k \ne 1$.
\\

The coefficient in the second double sum in~\eqref{eq:CancellationStart} is equal to 1 and we just need to consider all permutations of $n_k$ for $k \ge 2$ in order to obtain the desired result,
\begin{align}
(\sum_iE_i \rho E_i^{\dagger})_{pp}&=\left(\dfrac{4^{d-1}}{4^{d-1}(d+1)}\sum_i |\beta_i|^2+\dfrac{|\beta_p|^2}{d+1}\right)\ketbra{p}{p}
\label{eq:CancellationEnd} \\
&=\dfrac{1}{d+1}(|\beta_1|^2+|\beta_2|^2+\cdots+|\beta_{p-1}|^2+2|\beta_p|^2+|\beta_{p+1}|^2+\cdots+|\beta_d|^2).
\end{align}

Now consider the case $p \ne q$ in~\eqref{eq:KrausCheck},
\begin{align}
&(\sum_i E_i\rho E_i^{\dagger})_{pq}=\sum_{n_2,\cdots,n_d=0}^3 \left(\dfrac{ i^{n_p+n_q}(-1)^{n_q}}{4^{d-1}(d+1)}\sum_{k,l} i^{n_k+n_l}(-1)^{n_l}\beta_k \overline{\beta}_l\right)\ketbra{p}{q}\\
& \qquad=\sum_{n_2,\cdots,n_d=0}^3 \dfrac{ i^{n_p+n_q}(-1)^{n_q}}{4^{d-1}(d+1)} \left( i^{n_p+n_q}(-1)^{n_q}\beta_p \overline{\beta}_q+ i^{n_q+n_p}(-1)^{n_p}\beta_q \overline{\beta}_p + i^{n_p} \sum_{l \atop (l \ne q)}  (-i)^{n_l} \beta_p \overline{\beta}_l \right.   \nonumber \\
&  \qquad \qquad + \left.  i^{n_q} \sum_{l \atop (l \ne p)} (-i)^{n_l} \beta_q \overline{\beta}_l + (-i)^{n_p} \sum_{k \atop (k \ne q)} i^{n_k} \beta_k \overline{\beta}_p +  (-i)^{n_q} \sum_{k \atop (k \ne p)} i^{n_k} \beta_k \overline{\beta}_q \right. \nonumber \\
& \qquad \qquad \qquad \qquad  + \left. \sum_{k,l \atop k,l \notin \{p,q\}} i^{n_k+n_l}(-1)^{n_l}\beta_k \overline{\beta}_l \right) \ketbra{p}{q}
\end{align}
\begin{align}
&=\dfrac{1}{4^{d-1}(d+1)} \sum_{n_2,\cdots,n_d=0}^3 \left( i^{2(n_p+n_q)}(-1)^{2n_q}\beta_p \overline{\beta}_q+ i^{2(n_q+n_p)}(-1)^{n_p+n_q}\beta_q \overline{\beta}_p +  i^{2n_p} (-i)^{n_q} \sum_{l \atop (l \ne q)} (-i)^{n_l} \beta_p \overline{\beta}_l \right. \nonumber \\
& \qquad \qquad \qquad \qquad \qquad +  i^{n_p} \sum_{l \atop (l \ne p)} (-i)^{n_l} \beta_q \overline{\beta}_l + (-i)^{n_q} \sum_{k \atop (k \ne q)} i^{n_k} \beta_k \overline{\beta}_p + i^{n_p}(-i)^{2n_q} \sum_{k \atop (k \ne p)} i^{n_k} \beta_k \overline{\beta}_q \nonumber \\
& \qquad \qquad \qquad \qquad \qquad \qquad \left. + i^{n_p+n_q}(-1)^{n_q}\sum_{k,l \atop k,l \notin \{p,q\}} i^{n_k+n_l}(-1)^{n_l}\beta_k \overline{\beta}_l\right)\ketbra{p}{q}\\
&=\dfrac{1}{4^{d-1}(d+1)} \sum_{n_2,\cdots,n_d=0}^3 \left( (-1)^{n_p+n_q}\beta_p \overline{\beta}_q+\beta_q \overline{\beta}_p +  (-1)^{n_p} (-i)^{n_q} \sum_{l \atop (l \ne q)} (-i)^{n_l} \beta_p \overline{\beta}_l \right. \nonumber \\
& \qquad \qquad \qquad \qquad \qquad +  i^{n_p} \sum_{l \atop (l \ne p)} (-i)^{n_l} \beta_q \overline{\beta}_l + (-i)^{n_q} \sum_{k \atop (k \ne q)} i^{n_k} \beta_k \overline{\beta}_p + i^{n_p}(-1)^{n_q} \sum_{k \atop (k \ne p)} i^{n_k} \beta_k \overline{\beta}_q \nonumber \\
& \qquad \qquad \qquad \qquad \qquad \qquad \left. + i^{n_p+n_q}(-1)^{n_q}\sum_{k,l \atop k,l \notin \{p,q\}} i^{n_k+n_l}(-1)^{n_l}\beta_k \overline{\beta}_l\right)\ketbra{p}{q} \label{eq:Cancellation2}.
\end{align}
All terms in~\eqref{eq:Cancellation2} are zero except for the second term because we can permute the values of $n_p$ or $n_q$ in a similar way to the technique used in~\eqref{eq:CancellationStart}-\eqref{eq:CancellationEnd}. Finally there are $4^{d-1}$ copies of the second term since this term is independent of the permuted values of $n_k$ for $k \ge 2$. Thus
\begin{align}
(\sum_i E_i\rho E_i^{\dagger})_{pq}=\dfrac{1}{d+1}\beta_q \overline{\beta}_p,
\end{align}
which agrees with the output density matrix in~\eqref{eq:outputmatrix}.
\\

% Second part of proof

Now show~\eqref{eq:2}:
\begin{align}
\sum_i E_i^{\dagger}E_i &= \dfrac{1}{d+1}\left( \sum_j(\ketbra{j}{j})(\ketbra{j}{j}) \right. \nonumber \\
&\qquad \qquad \qquad \left. + \dfrac{1}{4^{d-1}} \sum_{n_2, \cdots, n_d=0}^3 (\sum_{k=1}^d i^{n_k}\ket{k})(\sum_{k'=1}^d i^{n_{k'}} \bra{k'})(\sum_{l=1}^d (-i)^{n_l}\ket{l})(\sum_{l'=1}^d (-i)^{n_{l'}} \bra{l'}) \right) \\
&= \dfrac{1}{d+1} \left( \sum_j \ketbra{j}{j} + \dfrac{1}{4^{d-1}} \sum_{n_2, \cdots, n_d=0}^3 (\sum_{k=1}^d i^{n_k}\ket{k})(\sum_{l=1}^d (1)^{n_l})(\sum_{l'=1}^d (-i)^{n_{l'}} \bra{l'}) \right) \\
&= \dfrac{1}{d+1} \left( \sum_j \ketbra{j}{j} + \dfrac{d}{4^{d-1}} \sum_{n_2, \cdots, n_d=0}^3 (\sum_{k=1}^d i^{n_k}\ket{k})(\sum_{l'=1}^d (-i)^{n_{l'}} \bra{l'}) \right) \label{eq:Cancellation3}
\end{align}
Now in the case when $k \ne l'$ in~\eqref{eq:Cancellation3}, one can apply the same technique as in \eqref{eq:CancellationStart}-\eqref{eq:CancellationEnd} to show that this last term is equal to zero. However when $k=l'$ one obtains
\begin{align}
\sum_i E_i^{\dagger}E_i &= \dfrac{1}{d+1} \left( \sum_j\ketbra{j}{j}+\dfrac{d}{4^{d-1}} \sum_{n_2,\cdots, n_d=0}^3 \sum_{k=1}^d  i^{n_k} (-i)^{n_k} \ketbra{k}{k} \right) \\
&= \dfrac{1}{d+1} \left( \sum_j \ketbra{j}{j} + \dfrac{d}{4^{d-1}} \sum_{n_2,\cdots, n_d=0}^3  \sum_{k} \ketbra{k}{k} \right) \\
&= \dfrac{1}{d+1} \left( \sum_j \ketbra{j}{j} + d  \sum_{j} \ketbra{j}{j} \right) = \sum_j \ketbra{j}{j}
\end{align}
which agrees with~\eqref{eq:2}.

\end{proof}

% \section*{References}
\bibliographystyle{unsrt}
\bibliography{bibtex}

\end{document}